\documentclass[12pt]{article}
\usepackage{amsthm,amsfonts,amsmath,amssymb,amscd}
\usepackage{graphicx}
\usepackage{mathtools}
\usepackage{enumerate}
\mathtoolsset{showonlyrefs}
\usepackage{verbatim}
\usepackage{authblk}
\usepackage[utf8]{inputenc}

\newcommand{\cA}{{\mathcal A}}

\newcommand{\cL}{{\mathcal L}}
\newcommand{\cC}{{\mathcal C}}
\newcommand{\cN}{{\mathcal N}}

\newtheorem{lemma}{Lemma}

\newtheorem{example}{Example}

\newtheorem{theorem}{Theorem}

\newtheorem{definition}{Definition}

\newcommand{\ket}[1]{\left\vert #1\right\rangle}
\newcommand{\bra}[1]{\left\langle #1\right\vert}

\newcommand{\Tr}{\mathsf{tr}}
\newcommand{\tr}{\mathsf{tr}}

\renewcommand{\vec}[1]{\boldsymbol{#1}}

\newcommand{\id}{\mathop{\rm id}}

\newcommand{\norm}[1]{\left\Vert #1 \right\Vert}

\newcommand{\boxendproof}{\hspace*{\fill}{{$\Box$}} \vspace{10pt}}

\newcommand{\be}{\begin{equation}}
\newcommand{\ee}{\end{equation}}
\newcommand{\bea}{\begin{eqnarray}}
\newcommand{\eea}{\end{eqnarray}}
\newcommand{\beann}{\begin{eqnarray*}}
\newcommand{\eeann}{\end{eqnarray*}}

\newcommand{\n}{{\bf n}}

\newcommand{\proj}[1]{|#1\rangle\langle#1|}
\newcommand*{\cLh}{\cL_{\mathsf{heat}}}

\usepackage{color}

\newcommand{\comm}[2]{\left[#1,#2\right]}

\newcommand*{\half}{\frac{1}{2}}

\begin{document}
\title{Geometric inequalities from phase space translations}
\author[1,2]{Stefan Huber \thanks{stefan.huber@ma.tum.de}}
\author[1,2]{Robert K\"onig}
\author[2,3]{Anna Vershynina}
\affil[1]{Institute for Advanced Study\\ Technische Universit\"at M\"unchen, 85748 Garching, Germany}
\affil[2]{Zentrum Mathematik\\
Technische Universit\"at M\"unchen, 85748 Garching, Germany}
\affil[3]{BCAM - Basque Center for Applied Mathematics, 48009 Bilbao, Spain}
\renewcommand\Authands{ and }
\renewcommand\Affilfont{\itshape\small}

\date{\today}
\maketitle
\begin{abstract}
We establish a quantum version of the classical isoperimetric inequality relating the Fisher information and the entropy power of a quantum state. The key tool is a Fisher information inequality for a state which results from a certain convolution operation: the latter maps a classical probability distribution on phase space and a quantum state to a quantum state. We show that this inequality also gives rise to several related inequalities whose counterparts are well-known in the classical setting: in particular, it implies an entropy power inequality for the mentioned convolution operation as well as the isoperimetric inequality, and establishes concavity of the entropy power along trajectories of the quantum heat diffusion semigroup. As an application, we derive a Log-Sobolev inequality for the quantum Ornstein-Uhlenbeck semigroup, and argue that it implies fast convergence towards the fixed point for a large class of initial states.
\end{abstract}

\section{Introduction}
The convolution operation $(X,Y)\mapsto X+Y$ between two 
real- (respectively vector-) valued independent random variables~$X$ and $Y$ plays a central role in classical information theory. The operation is defined in terms of the action on probability density functions as
\begin{align}
(f_X,f_Y)\mapsto f_{X+Y}\qquad\textrm{ where }\qquad f_{X+Y}(z):=\int f_X(z-x)f_Y(x)dx\ .\label{eq:classicalconvolution}
\end{align}  
The convolution models a general class of additive noise channels, and thus provides a natural framework for the study of information processing and associated capacities. 
The operation~\eqref{eq:classicalconvolution}  also is a central element in many functional analytic inequalities, most notably Young's inequality~\cite{Beckner75}, the Brascamp-Lieb inequalities~\cite{BrascampLieb76},  de Bruijn's identity~\cite{Stam59}, the Fisher information inequality~\cite{Stam59}, and the entropy power inequality~\cite{Shannon48-2,Blachman65}. Such inequalities have wide use in information theory, yielding bounds on communication capacities, as observed by Shannon~\cite{Shannon48-2}. They can also provide, for example, bounds on the convergence rate in the central limit theorem~\cite{Barron86central}. Beyond these applications, these results are appealing from a conceptual, geometric viewpoint: many inequalities can be regarded as
information-theoretic counterparts of related statements about convex bodies. For example, identifying entropy power with volumes reveals a formal similarity between the Brunn-Minkowski inequality and the entropy power inequality~\cite{CostaThomas84}. Indeed, there is even a proof  of the latter  guided by this intuition~\cite{SzarekVoicu00}. We refer to~\cite{Demboetal91,barthe06} for detailed accounts of this wealth of inequalities and their interrelationships.

Our work is guided by the question of whether a similar array of inequalities exists in a quantum setting. A key first step in this direction -- one which is directly relevant to our work --  was taken by Werner~\cite{WernerHarmonicanalysis84}. He introduced the convolution operation 
\begin{align}
(f,\rho)\mapsto f\star_t \rho\qquad\textrm{ where }\qquad f\star_t \rho=\int f(\xi) 
W({\sqrt{t}\xi}) \rho W({\sqrt{t}\xi)}^\dagger d\xi\ ,\label{eq:cqconvolution}
\end{align}
 which involves a probability density function~$f$ on phase space as well as a state~$\rho$ (of a bosonic system). The operation  results in the average state~$f\star_t \rho$  when displacing~$\rho$ according to~$f$ using the (Weyl) displacement operators~$W(\xi)$. Here we introduce the parameter~$t\geq 0$ for convenience, the case $t=1$ was considered in~\cite{WernerHarmonicanalysis84}.
 Treating the convolutions~\eqref{eq:classicalconvolution} and~\eqref{eq:cqconvolution} on the same algebraic footing, Werner established (among other results characterizing~\eqref{eq:cqconvolution}) a form of Young's inequality. It is worth mentioning that Carlen and Lieb have recently established generalizations of the latter in a fermionic context \cite{Carlenfermionic}.

More recent work~\cite{KoeSmiEPI} has centered around a convolution operation of the form
\begin{align}
(\rho_X,\rho_Y)\mapsto  \rho_{X\boxplus_\lambda Y}=\tr_2(U_\lambda (\rho_X\otimes \rho_Y)U_\lambda^\dagger)\ ,\label{eq:qqconvolution}
\end{align}
where $U_\lambda$ is a ($d$-mode) beamsplitter of transmissivity~$\lambda\in [0,1]$. It is worth pointing out that the action of this operation is formally similar to~\eqref{eq:classicalconvolution} when expressed in terms of the Wigner functions describing the  quantum states~$\rho_X$ and $\rho_Y$. The map~\eqref{eq:qqconvolution} describes a process where two states interact. It captures, in particular, the situation where one of the states is transmitted through an additive (bosonic) noise channel. In~\cite{KoeSmiEPI}, the authors established an entropy power inequality of the form
\begin{align}
e^{S(\rho_X\boxplus_{\lambda} \rho_Y)/d}&\geq \lambda e^{S(\rho_X)/d}+(1-\lambda) e^{S(\rho_Y)/d}\ ,\label{eq:epiqq}
\end{align}
for the convolution~\eqref{eq:qqconvolution} and for~$\lambda=1/2$, 
where $S(\rho_X)=-\tr(\rho_X\log\rho_X)$ denotes the von Neumann entropy. Subsequent work~\cite{MariPalma15} managed to lift the restriction on $\lambda$, and generalized this result to more general (Gaussian) unitaries in place of $U_\lambda$.   
A related inequality of the form $S(\rho_X\boxplus_\lambda \rho_Y)\geq \lambda S(\rho_X)+(1-\lambda)S(\rho_Y)$ for~$\lambda\in [0,1]$ was also shown in~\cite{KoeSmiEPI}, generalizing classical results~\cite{Lieb78} (see also~\cite{VerduGuo06} for a discussion of the relationship between the two). A generalization to conditional entropies was proposed in \cite{koenigconditionalepi2015}, and an application to channel capacities was discussed in~\cite{KoeSmiEPIChannel}. 

A key tool in establishing these results is the quantum Fisher information~$J(\rho)$,  
defined for a state~$\rho$ as the 
divergence-based Fisher information of the family $\left\{\rho^{(\theta)}\right\}_{\theta}$ obtained by displacing $\rho$ along a phase space direction (see Section~\ref{sec:quantumdef} for a precise definition). It was shown in~\cite{KoeSmiEPI} for $\lambda=1/2$ and in~\cite{MariPalma15} for general $\lambda\in [0,1]$ that this quantity satisfies the Fisher information  inequality~\cite{Stam59} 
\begin{align}
J(\rho_X\boxplus_\lambda \rho_Y)^{-1}\geq \lambda J(\rho_Y)^{-1}+(1-\lambda) J(\rho_X)^{-1}\ .\label{eq:fisherinformationinequality}
\end{align}
This identity is a consequence of the strong subadditivity (data processing) inequality for relative entropy, and lies at the heart of the proof of~\eqref{eq:epiqq}.

Following the theme of entropy power inequalities for quantum systems, Audenaert, Datta, and Ozols~\cite{datta15} have obtained strong majorization-type  results for the finite-dimensional case. These center around an operation of the form~\eqref{eq:qqconvolution}, but with~$\rho_X$ and $\rho_Y$ being states on a finite-dimensional Hilbert space, and a family $U_\lambda=e^{i\lambda H}$ of unitaries generated by a Hamiltonian~$H$  realizing the SWAP-operation of the two systems. 
As argued by Audenaert et al.,~these results imply several entropy-power-type inequalities.  Very recently, Carlen, Lieb, and Loss~\cite{Carlenetal16} have provided an elegant short proof of these statements. Guha, Shapiro, and Garc\'{i}a-Patr\'{o}n have discussed alternative definitions of the quantum entropy power and corresponding entropy power inequalities \cite{guha16}.

 \section{Our contribution}
\subsection{New geometric inequalities for bosonic systems}
Here we focus on the convolution operation~\eqref{eq:cqconvolution}. We find that this operation satisfies similar 
properties as the convolutions~\eqref{eq:classicalconvolution} and~\eqref{eq:qqconvolution}. In particular, we establish a Stam inequality for the Fisher information of the form
\begin{align}
J(f\star_t\rho)^{-1}\geq  J(\rho)^{-1}+t J(f)^{-1},\qquad\textrm{ for all }t\geq 0\ .\label{eq:stamcq}
\end{align} 
This classical-quantum version of Stam's inequality has several immediate consequences, all of which follow the reasoning used in establishing classical results~\cite{Demboetal91}. 
For example, we use \eqref{eq:stamcq} to establish an entropy power inequality of the form
\begin{align}
  \exp\left(S(f\star_t\rho)/d\right) \geq \exp\left(S(\rho)/d\right) + t \exp\left(H(f)/d\right)\ .
  \label{eq:entropypowerq}
\end{align}
Taking $f=f_Z$ to be a unit-variance centered Gaussian, we find a quantum isoperimetric inequality of the form
\begin{align}
\frac{d}{dt}\Big|_{t=0}\left(\frac{1}{2d}J(f_Z\star_t \rho)\right)^{-1}\geq 1\ .
\end{align}
Note that $f_Z\star_t \rho$ is the result of applying a classical noise channel to~$\rho$, where the variance of the displacement goes to~$0$ in the limit~$t\rightarrow 0$. This family of maps constitutes a semigroup, which we call
the heat diffusion semigroup: it is generated by a Liouvillian~$\cLh$ such that
\begin{align}
f_Z\star_t\rho=e^{t\cLh}(\rho)\qquad\textrm{ for an initial state }\rho\textrm{ and }t\geq 0\ .
\end{align}
We find that the entropy power along trajectories generated by this Liouvillian is concave, i.e.,
\begin{align}
\frac{d^2}{dt^2}\Big|_{t=0}\exp\left(\frac{1}{d}S(e^{t\cLh}(\rho))\right)\leq 0\ . \label{eq:qheatequationsemigroup}
\end{align}
Eq.~\eqref{eq:qheatequationsemigroup} generalizes a celebrated result~\cite{costa85,Villanishortproofconcavity} concerning the classical heat equation. The entropy power inequality \eqref{eq:entropypowerq} for Gaussian $f$ implies the lower bound $\exp\left(\frac{1}{d}S(e^{t\cLh}(\rho))\right)\geq~\exp\left(\frac{1}{d}S(\rho)\right) + (2\pi e)t$, and establishes the isoperimetric inequality for the Fisher information
\begin{align}
J(\rho)\exp\left(\frac{1}{d}S(\rho)\right)\geq 4\pi e\, d\ \label{eq:isoperimetryfisher}
\end{align}
for states~$\rho$ of $d$ bosonic modes. We find that for $d=1$, this statement is  tight: Gaussian thermal states achieve equality in~\eqref{eq:isoperimetryfisher} in the limit of large mean photon numbers. 

\subsection{Application to the Ornstein-Uhlenbeck semigroup}
\label{sec:appornstein}
We apply our results, in particular Eq.~\eqref{eq:isoperimetryfisher},  to the quantum Ornstein-Uhlenbeck (qOU) semigroup
for a one-mode bosonic system.
This is a one-parameter group of CPTP maps
$\{e^{t\cL_{\mu,\lambda}}\}_{t\geq 0}$ generated by a linear combination of  Liouvillians of a quantum amplifier and an attenuator channel, respectively:
\begin{align}
\cL_{\mu,\lambda}&=\mu^2\cL_-+\lambda^2\cL_+\qquad\textrm{ for } 0<\lambda<\mu\ ,\textrm{ where }\label{eq:qOUsemigroup}\\
 \cL_+(\rho)&= a^\dagger\rho a-\frac{1}{2}\{a a^\dagger, \rho\}\qquad\textrm{ and }\qquad  \cL_-(\rho)=a\rho a^\dagger-\frac{1}{2}\{a^\dagger a, \rho\}\ ,
\end{align}
where $a^\dagger$ and $a$ are the creation and annihilation operators (i.e., $[a,a^\dagger] = \mathsf{id}$).
The qOU semigroup~\eqref{eq:qOUsemigroup}
is a natural counterpart of the classical semigroup generated by the Fokker-Planck equation (see Appendix~\ref{sec:classicalOU}).
The unique fixed point of the semigroup~$\{e^{t\cL_{\mu,\lambda}}\}_{t\geq 0}$ is the state
\begin{align}
\sigma_{\mu,\lambda}&=(1-\nu)\sum_{n=0}^\infty \nu^n \proj{n}\qquad\textrm{ with }\qquad\nu=\lambda^2/\mu^2\ ,
\end{align}
i.e., it is diagonal in the number state basis $\{\ket{n}\}_{n\in\mathbb{N}_0}$ with a geometric distribution, hence a Gaussian thermal state. 

We conjecture that 
for an arbitrary initial state~$\rho$, this semigroup converges  to the fixed point at an exponential rate given by the exponent~$\mu^2-\lambda^2$, that is,
\begin{align}
D(e^{t\cL_{\mu,\lambda}}(\rho)\|\sigma_{\mu,\lambda})\leq e^{-(\mu^2-\lambda^2)t}D(\rho\|\sigma_{\mu,\lambda})\qquad\textrm{ for all }t\geq 0\ ,\label{eq:mainconjectureqou}
\end{align}
where $D(\rho\|\sigma) = \tr(\rho\log\rho - \rho\log\sigma)$ is the relative entropy.
In Appendix \ref{app:qOU}, we show that \eqref{eq:mainconjectureqou} holds for all Gaussian states~$\rho$, and the exponent $\mu^2-\lambda^2$ is optimal.
In other words, our conjecture
amounts to the statement that certain Gaussian thermal states converge ``most slowly'' to the fixed point, and the Log-Sobolev-$1$-constant, defined
as the largest constant $\alpha_1>0$ such that 
$D(e^{t\cL}(\rho)\|\sigma)\leq e^{-2\alpha_1 t}D(\rho\|\sigma)$
for any state $\rho$ and all $t\geq 0$, is given by $\alpha_1=\frac{1}{2}(\mu^2-\lambda^2)$. 

The quantum isoperimetric inequality~\eqref{eq:isoperimetryfisher} provides evidence for this conjecture: Taking as a specific example the case where $\lambda=1$ and $\mu=\sqrt{2}$ (and thus $\mu^2-\lambda^2=1$),
 we can show (see Example~\ref{ex:2-1_process_n} in Section~\ref{sec:app}) that
 \begin{align}
\frac{d}{dt}\Big|_{t=0}D(e^{t\cL_{\sqrt{2},1}}(\rho)\|\sigma_{\sqrt{2},1})\leq -D(\rho\|\sigma_{\sqrt{2},1})\qquad\textrm{ for all states }\rho\textrm{ with } \tr(\rho a^\dagger a)\lesssim 0.67\ .\label{eq:contractionratedef}
\end{align} 
We also show (see Example~\ref{ex:2-1_process_S} in Section~\ref{sec:app})  that recent work by de Palma, Trevisan, and Giovannetti~\cite{depalma2016} similarly implies 
exponential convergence of the form~\eqref{eq:mainconjectureqou}
for initial states~$\rho$ having large entropy: the isoperimetric inequality of~\cite{depalma2016} (which we discuss in more detail below) implies that 
\begin{align}
\frac{d}{dt}\Big|_{t=0}D(e^{t\cL_{\sqrt{2},1}}(\rho)\|\sigma_{\sqrt{2},1})\leq -D(\rho\|\sigma_{\sqrt{2},1})\qquad\textrm{ for all states }\rho\textrm{ with } S(\rho)\gtrsim\ 2.4\ .\label{eq:contractionratedef2}
\end{align}
While this does not establish the scaling~\eqref{eq:mainconjectureqou} for all states~$\rho$,  it illustrates the use of the quantum isoperimetric inequality in a concrete context. We clarify the relationship between our case and the one in the classical setting (see Section~\ref{sec:classicalcontext}), where
the Log-Sobolev-$1$ constant 
can be obtained from the isoperimetric inequality for the classical Fisher information, following work by Carlen~\cite{carlen94}. In fact, this is the main motivation for our conjecture: Gaussian distributions give the optimal convergence rate in the classical problem.

Whether conjecture~\eqref{eq:mainconjectureqou} holds is a challenging open question. We remark that it would significantly strengthen known results about the qOU semigroup. Specifically, Carbone et al.~\cite{carbonesasso07} established  that the qOU semigroup is hypercontractive. In~\cite[Proposition 4.2]{carbonesasso07}, the following 
inequality was shown for the Log-Sobolev-$2$ constant\footnote{Here we follow the notation of~\cite{olkiewiczzergarlinski,kastoryanotemme,hermesetal}. In contrast, this is denoted~$\alpha_1$ in~\cite{carbonesasso07}.}~$\alpha_2$ of $\cL_{\mu,\lambda}$: 
\begin{align}
\alpha_C^{-1}\leq \alpha_2^{-1}\leq \frac{4(5-\log(1-\nu))}{\mu^2(1-\nu)}+(3\log 3)\alpha_C^{-1}\ \qquad\textrm{ for }\qquad\nu=\lambda^2/\mu^2\ .
\end{align}
 In this expression, $\alpha_C$ is the Log-Sobolev-$2$ constant 
of the associated classical birth-and-death process (which is unknown), but for which  the bounds 
\begin{align}
\frac{\log \nu^{-1}}{5\sqrt{5}\mu^2(1-\nu)^{3/2}}\leq \alpha_C^{-1}\leq 
\frac{255}{4}\frac{(1+\log 2)(1-\nu)+\log \nu^{-1}}{\mu^2(1-\nu)^3}
\end{align}
were shown (see~\cite[Proposition 4.1]{carbonesasso07}
where $\alpha_C$ is  denoted $\alpha_0$).
Following~\cite{olkiewiczzergarlinski,kastoryanotemme,hermesetal},
this implies that for all states~$\rho$, we have~$D(e^{t\cL_{\mu,\lambda}}(\rho)\|\sigma_{\mu,\lambda})\leq e^{-2\alpha_2t}D(\rho\|\sigma)$ (respectively, we have $D(e^{t\cL_{\mu,\lambda}}(\rho)\|\sigma_{\mu,\lambda})\leq e^{-\alpha_2t}D(\rho\|\sigma_{\mu,\lambda})$), if the semigroup can be shown to be strongly (respectively weakly) $L_p$-regular.

The derivation of our fast convergence results from
the isoperimetric inequality~\eqref{eq:isoperimetryfisher} follows, to some extent, a well-known  line of reasoning considered in the classical context. Indeed, Carlen~\cite{carlen94} has shown that the isoperimetric inequality gives rise to a Log-Sobolev inequality, which in turns provides bounds on the convergence of the classical Ornstein-Uhlenbeck (cOU) semigroup to the fixed point (see Appendix~\ref{sec:classicalOU}).  However, in our case we find that, while a Log-Sobolev inequality again easily follows from~\eqref{eq:isoperimetryfisher}, the quantities appearing in it are not easily estimated. This concerns, in particular, the entropy rate  of the attenuator
\begin{align}
J_-(\rho)&=2\frac{d}{dt}\Big|_{t=0}  S(e^{t\cL_-}(\rho))\ ,
\end{align}
given an arbitrary initial state~$\rho$ (the factor~$2$ is chosen for convenience only). This is in sharp contrast to the classical case: here
the trivial identity $H(e^{t}X)=H(X)+t$, for a scalar~$t>0$, satisfied by the differential entropy~$H(X)$ of a random variable~$X$, is sufficient for the purpose of establishing fast convergence of the cOU semigroup to the fixed point (see Appendix~\ref{sec:classicalOU}). 

De Palma et al.~\cite{depalma2016} showed that the infimum 
$\inf_{\rho:\, S(\rho)=S}J_-(\rho)$ over all states~$\rho$ with a given entropy is achieved by a Gaussian thermal state. This statement, combined with
a lower bound on the corresponding quantity~$J_+(\rho)$ for the amplifier from~\cite{buscemiwilde}, valid for all states~$\rho$, immediately yields inequality~\eqref{eq:contractionratedef2}.

To establish \eqref{eq:contractionratedef}, we prove another lower bound on~$J_-(\rho)$: more precisely, we show
that the infimum 
\begin{align}
\inf_{\rho:\, \tr(\rho a^\dagger a)\leq\n }J_-(\rho)
\end{align}
over all states~$\rho$ with mean photon number bounded by~$\n$
is achieved by a Gaussian thermal state. The proof proceeds by reduction to the classical case using recent majorization-type results~\cite{jabbouretal,depalmaetal}. The latter can be treated using the results from \cite{depalma2016}. This, then, provides the required lower bound and yields statement~\eqref{eq:contractionratedef}. It is not, however, tight enough to establish Conjecture~\eqref{eq:mainconjectureqou}.

Understanding the relationship between  entropy production rates along trajectories of the qOU semigroup, i.e., different quantities of the form
$J_{\mu,\lambda}(\rho) = 2\frac{d}{dt}\Big|_{t=0}  S(e^{t\cL_{\mu,\lambda}}(\rho))$ for $(\mu,\lambda)\neq (1,0)$,
remains an open problem. We believe that progress in this direction could help provide further evidence for (or indeed lead to a proof of) the validity of conjecture~\eqref{eq:mainconjectureqou}.

\subsection{Prior work in the classical setting\label{sec:classicalcontext}}
Our work in the quantum setting follows a long sequence of well-known existing arguments applicable to classical probability distributions. All geometric inequalities established here have classical counterparts, and their proofs are inspired by (and directly generalize)  corresponding classical proofs. This raises the question of whether other analogues of classical results exist: for example, one may conjecture that  there is a quantum counterpart of Young's inequality for the convolution operation~\eqref{eq:qqconvolution}. 

It is not our intention to provide a complete review of this assortment of classical results: it is hardly possible to do justice to the many important developments in this area. We refer to the  article~\cite{Demboetal91} for a survey of many known connections. Instead, we briefly review some of the basic definitions
and seminal results which are directly relevant to our work.

In the classical setting we assume to have $\mathbb{R}^d$-valued random variables with absolutely continuous density functions. For such a random variable~$X$ with density function~$f$ (which we often assume to be non-vanishing everywhere for simplicity), a fundamental information measure of interest is the \textit{Shannon (differential) entropy}
\begin{align}
H(f)&=-\int_{\mathbb{R}^d} f(\vec{x})\log f(\vec{x})d\vec{x}\ ,
\end{align}
also denoted as $H(X)$. The {\em entropy power} is given by
\begin{align}
  N(f) &= \exp\left({2H(f)/d}\right)\ ,
\end{align}
and also written as $N(X)$.
Up to a factor, this quantity coincides with the variance of a Gaussian distribution having the same entropy. 
The \textit{divergence}, or \textit{relative entropy}, of two density functions $f, g$ is defined as
\begin{align}
  D\left(\left.f \right\| g\right) = \int_{\mathbb{R}^d} f(\vec{x}) \log \frac{f(\vec{x})}{g(\vec{x})}\;d\vec{x}\ .
  \label{def:divergence}
\end{align}
The \textit{Fisher information} of a random variable $X$ with density function $f$ is defined as the following quantity\footnote{In more general terms,~\eqref{eq:fisherinformationtranslation} is the trace of the Fisher information matrix
\begin{align}
\left( \left. \frac{\partial^2}{\partial\theta_i\partial\theta_j}\right|_{\vec{\theta} = \vec{\theta_0}} D\left( \left. f^{\left(\vec{\theta}_0\right)}\right\| f^{(\vec{\theta})}\right)\right)_{i,j=1}^{m}
\label{eq:def:_class_Fisher}
\end{align}
corresponding to~\eqref{eq:translateclassicaldistribution}.  The Fisher information matrix provides a lower bound on the variance of an unbiased estimate of the unknown parameter~$\vec{\theta}$ according to the Cramer-Rao inequality. For our purposes, it is sufficient to consider the family~\eqref{eq:translateclassicaldistribution}.}
\begin{align}
  J(f)&=\int \left(\nabla f(\vec{x})\right)^T\cdot \left(\nabla f(\vec{x})\right) \cdot \frac{1}{f(\vec{x})} d\vec{x} \ ,\label{eq:fisherinformationtranslation}
\end{align}
which is associated with the family of translated probability density functions
\begin{align}
f^{(\vec{\theta})}(\vec{x})=f(\vec{x}-\vec{\theta})\qquad\textrm{ for }\vec{\theta}\in\mathbb{R}^d\ .\label{eq:translateclassicaldistribution}
\end{align}
The quantity $J(f)$ is also often denoted as $J(X)$, to emphasize that it is the Fisher information of a random variable $X$ which has density function $f$.

For two densities~$f$ and $g$ describing random variables~$X$ and~$Y$, the convolution operation of interest is given by~\eqref{eq:classicalconvolution}, and describes the addition of the two random variables. Of particular interest is the case where one of the random variables is a centered normal distribution with unit variance. Such a random variable is denoted as~$Z$ below.  Key results in this setting are
\begin{description}
\item[the classical de Bruijn identity:]
  \begin{equation}
    \label{eq:de_Bruijn_cl}
   \frac{\partial}{\partial t}H \left(X + \sqrt{t}Z\right) = \frac{1}{2} J\left( X + \sqrt{t}Z\right)\ .
 \end{equation}
 This result was established by de Bruijn and gives an important relation between the Fisher information and the entropy when a random variable $X$ is perturbed under an additive Gaussian noise channel. It is a key ingredient in proofs of many information-theoretic inequalities. A simple proof can be found in \cite{CostaThomas84}.
\item[the Fisher information inequality:]
\begin{align}
  J\left( \sqrt{\lambda} X + \sqrt{1-\lambda}Y\right)\leq \lambda J(X) + (1-\lambda) J(Y),\ \mathrm{for}\ \lambda \in \left[0,1\right],
  \label{eq:fisherinformationineq}
\end{align}
as well as the related inequality
\begin{align}
  J\left(X + Y\right)^{-1} - J(X)^{-1} - J(Y)^{-1}\geq 0\ .
\end{align}
Proofs of these inequalities are given in \cite{Blachman65}, \cite{CostaThomas84}, and \cite{Zamir98}. Zamir~\cite{Zamir98} gives a particularly useful proof which relies on the information-processing inequality. This inequality states that the application of a channel cannot increase the Fisher information. Zamir's proof of the Fisher information inequality can be generalized to the quantum case.

\item[the Fisher information isoperimetric inequality]
  \begin{align}
    \frac{d}{d\epsilon}\bigg|_{\epsilon = 0}\left[\frac{1}{d}J(X+\sqrt{\epsilon}Z) \right]^{-1}\geq 1\ .
  \end{align}
This inequality implies that for Gaussian states, the inverse of the Fisher information has minimal sensitivity to additive Gaussian noise~\cite{dembosimpleconcavity}.

\item[the concavity of the entropy power]
  \begin{align}
    \frac{d^2}{d\epsilon^2}\bigg|_{\epsilon = 0} N(X + \sqrt{\epsilon}Z) \leq 0\ .
  \end{align}
This celebrated result establishes that the entropy power is a concave function along trajectories of the heat flow semigroup.
A proof is given in \cite{costa85}, and some shorter ones are presented in \cite{dembosimpleconcavity,Villanishortproofconcavity}.

\item[the entropy power inequality]
\begin{align} 
 N(X + Y) \geq N(X) + N(Y)\ .
\end{align}
Stam \cite{Stam59} gave a proof of the entropy power inequality which relies on the de Bruijn identity and the Fisher information inequality. The proof was later simplified by Blachman \cite{Blachman65} and others \cite{Barron86central,Demboetal91}.

\item[the isoperimetric inequality for entropies]
\begin{align}
  \frac{1}{d} J(X)N(X) \geq 2\pi e\ ,
\end{align}
as given in \cite{CostaThomas84}. The isoperimetric inequality for entropies implies that Gaussians have minimal entropy power among random variables with fixed Fisher information, and can be used to derive Log-Sobolev inequalities for the Ornstein-Uhlenbeck semigroup \cite{carlen94}, as we review in Appendix \ref{sec:classicalOU}.

\end{description}

\section{Preliminaries}\label{sec:quantumdef}

\subsection{States and information measures of interest}
We consider a $d$-mode bosonic system with
``position'' and ``momentum'' operators $(Q_k,P_k)$ of the~$k$-th mode satisfying  the canonical commutation relations~$[Q_j, P_k]=i\delta_{j,k}I$. Denoting the vector of position- and momentum-operators by $\vec{R} = \left(Q_1, P_1, \dots, Q_d, P_d\right)$, the Weyl displacement operators are defined as 
\begin{align}\label{eq:def:Weyl_op}
  W(\vec{\xi}) = e^{i \sqrt{2\pi}\, \vec{\xi} \cdot (\sigma\vec{R})}\qquad\textrm{ for } \vec{\xi} \in \mathbb{R}^{2d}\ .
\end{align}
Here  $\sigma = \begin{pmatrix}0 & 1 \\ -1 & 0 \end{pmatrix}^{\oplus d}$ is the matrix defining the symplectic inner  product. The factor~$\sqrt{2\pi}$
in the definition~\eqref{eq:def:Weyl_op} is for convenience only. From the commutation relations of position and momentum operators and the Campbell-Baker-Hausdorff formula, it is straightforward to check that the Weyl operators satisfy 
\begin{align}
  W(\vec{\xi})W(\vec{\eta}) = e^{-{i} \pi \vec{\xi} \cdot (\sigma \vec{\eta})}W(\vec{\xi} + \vec{\eta})\qquad\textrm{  for }\vec{\xi}, \vec{\eta} \in \mathbb{R}^{2d}\ .\label{eq:cbhcomposition}
\end{align}
Consider a state~$\rho$ on $d$~modes.  Quantities of interest are the von Neumann entropy $S(\rho)=-\tr(\rho\log \rho)$, as well as the relative entropy~$D(\rho\|\sigma)=\Tr \left(\rho \log \rho - \rho \log \sigma \right)$. The latter expression is defined for  positive operators~$\rho, \sigma$, and we will assume without further comments that the states~$\rho$,~$\sigma$ have full rank. 

For a multi-parameter family $\{\rho^{(\vec{\theta})}\}_{\vec{\theta}\in\mathbb{R}^{D}}$ of states depending smoothly on the parameters~$\vec{\theta}$, the divergence-based quantum Fisher information is defined as the trace of the Fisher information matrix \begin{align}
  J\left. \left( \{\rho^{(\vec{\theta})}\};\vec{\theta}\right)\right|_{\vec{\theta} = \vec{\theta_0}} = \left(\left. \frac{\partial^2}{\partial \theta_j \partial \theta_k} \right|_{\vec{\theta} = \vec{\theta_0}} D\left(\left. \rho^{(\vec{\theta_0})} \right\| \rho^{(\vec{\theta})}\right) \right)_{j,k = 1}^{D}\ .
\end{align}
This definition quantifies the dependence of the states on the parameter~$\vec{\theta}$  in the neighborhood of~$\vec{\theta}=\vec{\theta}_0$. 

In the following, we will apply this definition to the family~$\{ \rho^{(\vec{\theta})} \}_{\vec{\theta}\in\mathbb{R}^{2d}}$ of states
obtained by translating a given $d$-mode state~$\rho$:  Analogously to~\eqref{eq:translateclassicaldistribution}, where we translated a given probability density function~$f$, we define the translated states
 \begin{align}
  \rho^{(\vec{\theta})}:=W(\vec{\theta})\rho W(\vec{\theta})^\dagger\qquad\textrm{ for }\vec{\theta}\in\mathbb{R}^{2d}\ .
  \label{eq:translatedstate}
 \end{align}
Here translation by the parameter~$\vec{\theta}$ on phase space is achieved by means of the Weyl operators. The corresponding quantity
\begin{align}\label{eq:def:q_Fisher}
  J \left( \rho\right) = \Tr \left( \left. J\left(\{\rho^{(\vec{\theta})}\}; \vec{\theta}\right)\right|_{\vec{\theta}=\vec{0}} \right)\ ,
\end{align}
will simply be called the Fisher information of~$\rho$. 
Note that this definition matches that of the classical Fisher information~(cf.~\eqref{eq:def:_class_Fisher}). We emphasize that the concept of quantum Fisher information is non-unique (see~\cite{petzghi10}), but the use of~\eqref{eq:def:q_Fisher} is sufficient for our purposes.

\subsection{The quantum diffusion semigroup and the de Brujin identity}
Consider the Liouvillian defined on $d$ modes as
\begin{align}\label{eq:def:diffusion}
  \cLh(\rho) = - \pi \sum_{j=1}^{2d} \comm{R_j}{\comm{R_j}{\rho}}\ .
\end{align}
(The factor~$\pi$ differs from the convention used in~\cite{KoeSmiEPI}, but turns out to be convenient in the proof of Theorem \ref{thm:entropy_power}, as explained later.)
The one-parameter semigroup~$\{e^{t\cLh}\}_{t\geq 0}$ of completely positive trace-preserving maps (CPTPMs) generated by~$\cLh$ will be called the quantum (heat) diffusion semigroup. It has various nice properties: for example, as shown in \cite{KoeSmiEPI}, a quantum version of the de Bruijn identity~\eqref{eq:de_Bruijn_cl} reads
\begin{align}\label{eq:de_Bruijn_heat}
 \frac{d}{dt}\Big|_{t = 0} S\left( e^{t\cLh}(\rho)\right) = \half J(\rho)\ .
\end{align}
We remark that the proof of~\eqref{eq:de_Bruijn_heat}, which has  subsequently been applied, e.g., in~\cite{guha16} and generalized in~\cite{MariPalma15}, involves certain formal manipulations whose rigorous justification remains an interesting mathematical problem: as a quantum counterpart of partial integration, for example,
arguments under the trace need to be cyclically permuted. It is clear that such manipulations should be valid for sufficiently regular families of states (and indeed, are established for Gaussian states), but corresponding conditions are currently unknown.  We believe that the recent introduction of Schwartz operators in~\cite{keylkiukaswerner15} provides an appropriate framework to shed light on this aspect. In the following, we will assume that our states under consideration satisfy the required regularity assumptions. We hope 
that this issue will eventually  be resolved in a similar manner as in the classical setting, where initial work by Shannon~\cite{Shannon48-2} was followed by a long  sequence of papers with increasing rigor. In the case of the classical de Bruijn identity~\eqref{eq:de_Bruijn_cl},  Barron~\cite{Barron86central}, based on Stam's work~\cite{Stam59}, has shown validity for all random variables~$X$ with finite variance.

The map $e^{t\cLh}$ can be explicitly written as (see \cite{KoeSmiEPI,eisertwolfb})
\begin{align}
  e^{t\cLh}(\rho)   = \frac{1}{(2\pi)^d} \int e^{-\norm{\vec{\xi}}^2/2} W(\sqrt{t}\vec{\xi})\,\rho \, W(\sqrt{t}\vec{\xi})^\dagger d\vec{\xi}\ .\label{eq:heatequationintegrated} 
\end{align}
We may interpret this as the result of applying a certain convolution operation to a Gaussian distribution and a quantum state. More precisely, 
for $t\geq 0$, we  define a convolution operation~$\star_t$ between a probability density function $f : \mathbb{R}^{2d} \rightarrow \mathbb{R}$ and a $d$-mode state $\rho$ by
\begin{align}
  (f,\rho) \mapsto f \star_t \rho := \int f(\vec{\xi}) W(\sqrt{t}\vec{\xi})\, \rho\, W(\sqrt{t}\vec{\xi})^\dagger d\vec{\xi}\ .\label{eq:cF}
 \end{align}
In this terminology, Eq.~\eqref{eq:heatequationintegrated}  becomes 
 \begin{align}
   f_Z \star_t \rho = e^{t\cLh}(\rho)\ ,
 \end{align}
where $f_Z$ is a 
unit-variance centered Gaussian distribution, that is,
 \begin{align}\label{eq:def:Gaussian}
   f_Z (\vec{\xi}) = (2\pi)^{-d} e^{-\left\| \vec{\xi}\right\|^2/2}\ .
 \end{align}

 To close this section, we list two elementary properties of the convolution \eqref{eq:cF} which can be checked by straightforward calculation. If $f = f_X$ is the probability density function of a random variable $X$ and $t \geq 0$, then
 \begin{align}
   f_X\star_t \rho&=f_{\sqrt{t}X} \star_1\rho\ ,
   \label{eq:conv_scaling}
 \end{align}
 where the probability density function of the rescaled random variable $\sqrt{t}X$ is given by $f_{\sqrt{t}X}(\vec{\xi}) = f(\vec{\xi}/\sqrt{t})/\sqrt{t}^{2d}$. Addition of random variables corresponds to convolution in the following sense:

\begin{align}
  f_{X_1}\star_1 (f_{X_2}\star_1\rho)=f_{X_1+X_2}\star_1\rho\ ,
  \label{eq:conv_addition}
\end{align}
where $f_{X_1 + X_2}$ is defined as in \eqref{eq:classicalconvolution}.

\section{Quantum geometric inequalities}\label{sec:quantumineq}

In this section, we present several statements about the convolution operation \eqref{eq:cF} and the quantum Fisher information~\eqref{eq:def:q_Fisher}.
The key idea in establishing these results is the fact that
the convolution operation~\eqref{eq:cF} constitutes data processing, and hence provides an inequality because of the monotonicity of relative entropy. This is expressed in the following lemma. In the classical setting, the analogous argument for obtaining the Fisher information inequality~\eqref{eq:fisherinformationineq} was first emphasized by Zamir~\cite{Zamir98}. 
  
  \begin{lemma}{(Data processing inequality for convolution)}\label{lem:dataprocv}
  Let $f, g : \mathbb{R}^{2d} \rightarrow \mathbb{R}$ be probability density functions with full support. Then
\begin{align}\label{eq:q_data_processing}
D\left( \left. f \star_t \rho \right\| g \star_t \sigma \right) \leq D\left(\left. f \right\| g \right) + D\left( \left. \rho \right\| \sigma \right)\ 
\end{align}
for any states $\rho, \sigma$.
\end{lemma}
\begin{proof}
The quantum relative entropy satisfies the following scaling property for scalars~$\lambda, \mu > 0$:
\begin{align}\label{eq:divergencescale}
  D\left(\left.\lambda \rho \right\| \mu \sigma \right) = \lambda D\left( \left. \rho \right\| \sigma \right) - \lambda \Tr\left(\rho\right) \log \frac{\mu}{\lambda}\ .
\end{align} 
Defining $h(\vec{\xi}) = \frac{g(\vec{\xi})}{f(\vec{\xi})}$, we obtain (using the translated states defined in Eq. \eqref{eq:translatedstate})
  \begin{align}
D\left(\left. f \star_t \rho \right\| g \star_t \sigma \right) &= D\left( \left. f \star_t \rho \right\| \left( f \cdot h\right) \star_t \sigma \right)\\
&= D\left( \left. \int f(\vec{\xi}) \rho^{(\sqrt{t}\vec{\xi})}\; \mathrm{d}\vec{\xi}\right\|\int f(\vec{\xi}) h(\vec{\xi}) \sigma^{(\sqrt{t}\vec{\xi})}\;\mathrm{d}\vec{\xi} \right)\\
&\leq \int f(\vec{\xi}) D\left(\left. \rho^{(\sqrt{t}\vec{\xi})} \right\| h(\vec{\xi}) \sigma^{(\sqrt{t}\xi)}\right) \mathrm{d}\vec{\xi}\\
&= \int f(\vec{\xi})\left( D\left( \left. \rho^{(\sqrt{t}\vec{\xi})} \right\| \sigma^{(\sqrt{t}\vec{\xi})} \right) - \Tr\left(\rho^{(\sqrt{t}\vec{\xi})}\right) \log h(\vec{\xi})\right) \mathrm{d}\vec{\xi}\\
&= D\left( \left. \rho \right\| \sigma \right) - \int f(\vec{\xi}) \log \frac{g(\vec{\xi})}{f(\vec{\xi})}\\
&= D\left( \left. \rho \right\| \sigma \right)+D\left(\left. f \right\| g \right)\ .
\end{align}
Here the inequality we used is the joint convexity of the relative entropy (see~\cite[Theorem 1]{lieb73}). 
For the third equality we used property \eqref{eq:divergencescale}, and for the fourth equality we used unitary invariance of the relative entropy and the trace as well as the fact that $f$ is a probability distribution. The last equality follows from the definition~\eqref{def:divergence} of the divergence.
 \end{proof}
To convert Lemma~\ref{lem:dataprocv} into a statement about Fisher information, we need the following covariance property of the convolution operation~\eqref{eq:cF}: it breaks down translations of the state~$f\star_t\rho$ into translations of the function~$f$ and the state~$\rho$, respectively.

 \begin{lemma}\label{lemma:equiv} 
   Let $\omega_q, \omega_c>0$ and $t\geq 0$. Then
  \begin{align}\label{eq:equivalence}
    \left( f \star_t \rho \right)^{(\omega\vec{\theta})} =  f^{(\omega_c\vec{\theta})} \star_t \rho^{(\omega_q\vec{\theta})}\qquad\textrm{ for all } \vec{\theta } \in \mathbb{R}^{2d}\ ,
  \end{align}
  where $\omega=\omega_q+\sqrt{t}\omega_c$.
   \end{lemma}  
  \begin{proof}
 According to Definition~\eqref{eq:cF} and~\eqref{eq:cbhcomposition} we have 
 \begin{align}
    \left( f \star_t \rho \right)^{(\omega\vec{\theta})} &=\int f(\vec{\xi})W(\omega\vec{\theta})W(\sqrt{t}\vec{\xi})
   \rho W(\sqrt{t}\vec{\xi})^\dagger W(\omega\vec{\theta})^\dagger d\vec{\xi}\nonumber\\
   &=\int f(\vec{\xi})W(\omega\vec{\theta}+\sqrt{t}\vec{\xi})\rho W(\omega\vec{\theta}+\sqrt{t}\vec{\xi})^\dagger d\vec{\xi}\ .\label{inside:left}
 \end{align}
On the other hand, we similarly have 
\begin{align}
  f^{(\omega_c \vec{\theta})} \star_t \rho^{(\omega_q \vec{\theta})}
 &=
 \int f(\vec{\xi}-\omega_c\vec{\theta})W(\sqrt{t}\vec{\xi})W(\omega_q\vec{\theta})\rho W(\omega_q\vec{\theta})^\dagger W(\sqrt{t}\vec{\xi})^\dagger
 d\vec{\xi}\nonumber\\
 &=\int f(\vec{\xi}-\omega_c\vec{\theta})W(\omega_q \vec{\theta} + \sqrt{t}\vec{\xi})\rho W(\omega_q\vec{\theta}+\sqrt{t}\vec{\xi})^\dagger
 d\vec{\xi}\label{inside:right}\ .
\end{align}
Through a simple change of variables and recalling that $\omega=\omega_q+\sqrt{t}\omega_c$, the claim~\eqref{eq:equivalence} follows from~\eqref{inside:left} and~\eqref{inside:right}.
  \end{proof}

 Combining the data processing inequality of Lemma~\ref{lem:dataprocv} with Lemma~\ref{lemma:equiv}, we prove an inequality which may be seen as a classical-quantum version of the Stam inequality. It relates
the Fisher information of the state $f\star_t \rho$ to the Fisher informations of $f$ and $\rho$, respectively.
  \begin{theorem}\label{thm:InfoIneq} (Quantum Stam inequality)
    Let $\omega_q,\omega_c\in\mathbb{R},$ and $t\geq 0$. Then
  \begin{align}
    \omega^2 J(f \star_t \rho)\leq \omega_q^2 J(\rho)+\omega_c^2 J(f)\ ,
    \label{eq:fisherinfoineq}
  \end{align}
  where $\omega=\omega_q+\sqrt{t}\omega_c$. In particular,
  \begin{align}\label{eq:Stam}
    J(f \star_t \rho)^{-1}- J(\rho)^{-1}-tJ(f)^{-1}\geq 0\ .
  \end{align}
  \end{theorem}

  \begin{proof} Let $\vec{\theta_0}=(\theta_0^{(1)},\dots,\theta_0^{(2d)})\in\mathbb{R}^{2d}$. For $j=1,\dots 2d$ and $\theta_j \in \mathbb{R}$, introduce the vector 
    \begin{align}
      \vec{\tilde{\theta_j}} = \vec{\tilde{\theta_j}}(\theta_j) = (\theta_0^{(1)},\dots,\theta_0^{(j-1)}, \theta_j, \theta_0^{(j+1)},\dots,\theta_0^{(2d)})\ .
\end{align} Define the functions
\begin{equation}
  f(\theta_j):=D\left(\left. f^{(\omega_c\vec{\theta_0})} \right\| f^{(\omega_c\vec{\tilde{\theta_j}})}\right) + D\left( \left. \rho^{(\omega_q\vec{\theta_0})} \right\| \rho^{(\omega_q\vec{\tilde{\theta_j}})} \right)\ ,
\end{equation}
and 
\begin{equation}
  g(\theta_j):=D\left( \left. f^{(\omega_c\vec{\theta_0})} \star_t \rho^{(\omega_q\vec{\theta_0})} \right\| f^{(\omega_c\vec{\tilde{\theta_j}})} \star_t \rho^{(\omega_q\vec{\tilde{\theta_j}})} \right)\ .
\end{equation}
From the definition of the relative entropy and the data processing inequality \eqref{eq:q_data_processing}, for every $\theta_j$, we have
\begin{align}
&0\leq g(\theta_j)\leq f(\theta_j)\\
&0=f(\theta_0^{(j)})=g(\theta_0^{(j)}) \ .
\end{align}
The second derivative of $g$ can be written as the limit
\begin{equation}
\left. \frac{d^2}{d\theta_j^2}\right|_{\theta_j=\theta_0^{(j)}} g(\theta_j)=\lim_{\epsilon\rightarrow 0}\frac{g(\theta_0^{(j)}+\epsilon)-2g(\theta_0^{(j)})+g(\theta_0^{(j)}-\epsilon)}{\epsilon^2}\ ,
\end{equation}
and, therefore, is it bounded
\begin{equation}
0\leq \left. \frac{d^2}{d\theta_j^2}\right|_{\theta_j=\theta_0^{(j)}} g(\theta_j)\leq \left. \frac{d^2}{d\theta_j^2}\right|_{\theta_j=\theta_0^{(j)}} f(\theta_j)\ .
\end{equation}
Since 
\begin{align}
\Tr\left(J(\{f^{(\omega_c\vec{\theta})}\};\vec{\theta})\big|_{\vec{\theta}=\vec{\theta_0}}\right) + \Tr\left(J(\{\rho^{(\omega_q\vec{\theta})}\};\vec{\theta})\big|_{\vec{\theta}=\vec{\theta_0}}\right)=\sum_{j=1}^{2d}\left. \frac{d^2}{d\theta_j^2}\right|_{\theta_j=\theta_0^{(j)}} f(\theta_j)\ ,
\end{align}
and
\begin{align}
\Tr \left( J(\{f^{(\omega_c\vec{\theta})}\star_t \rho^{(\omega_q\vec{\theta})}\};\vec{\theta})\big|_{\vec{\theta}=\vec{\theta_0}} \right)=\sum_{j=1}^{2d}\left. \frac{d^2}{d\theta_j^2}\right|_{\theta_j=\theta_0^{(j)}} g(\theta_j)\ ,
\end{align}
we conclude that
\begin{equation}\label{eq:traces}
\Tr \left( J(\{f^{(\omega_c\vec{\theta})}\star_t \rho^{(\omega_q\vec{\theta})}\};\vec{\theta})\big|_{\vec{\theta}=\vec{\theta_0}} \right)\leq \Tr\left(J(\{f^{(\omega_c\vec{\theta})}\};\vec{\theta})\big|_{\vec{\theta}=\vec{\theta_0}}\right) + \Tr\left(J(\{\rho^{(\omega_q\vec{\theta})}\};\vec{\theta})\big|_{\vec{\theta}=\vec{\theta_0}}\right)\ .
\end{equation}
We remark that above inequality can also be derived as a matrix inequality without tracing both sides. However, this is not required for our purposes and our definition of the Fisher information.

Using Lemma~\ref{lemma:equiv}, the left-hand side of this equation for $\vec{\theta_0}=\vec{0}$ can be written as
 \begin{align}
   \Tr \left( J(\{f^{(\omega_c\vec{\theta})}\star_t \rho^{(\omega_q\vec{\theta})}\};\vec{\theta})\big|_{\vec{\theta}=\vec{0}} \right)&=
   \Tr \left( J(\{(f \star_t \rho)^{(\omega\vec{\theta})}\};\vec{\theta})\big|_{\vec{\theta}=\vec{0}} \right)=\omega^2 J(f \star_t \rho)\ .  \label{eq:xequationv}
 \end{align}
The right-hand side of Eq. \eqref{eq:traces} can be simplified by noticing that both the classical and quantum Fisher information matrices satisfy reparametrization formulas (see~\cite[Lemma~IV.1]{KoeSmiEPI}) 
\begin{align}
J(\{f^{(\omega_c\vec{\theta})}\};\vec{\theta})|_{\vec{\theta}=\vec{\theta_0}}=\omega_c^2J(\{f^{(\vec{\theta})}\};\vec{\theta})\qquad\text{ and }\qquad J(\{\rho^{(\omega_q\vec{\theta})}\};\vec{\theta})=\omega_q^2J(\{\rho^{(\vec{\theta})}\};\vec{\theta})\ .
\label{eq:reparamform}
\end{align}
Therefore, taking $\vec{\theta}=\vec{0}$ leads to
 \begin{align}
   \Tr\left(J(\{f^{(\omega_c\vec{\theta})}\};\vec{\theta})\big|_{\vec{\theta}=\vec{0}}\right) + \Tr\left(J(\{\rho^{(\omega_q\vec{\theta})}\};\vec{\theta})\big|_{\vec{\theta}=\vec{0}}\right)=\omega_c^2J(f)+
   \omega_q^2 J(\rho)\ .\label{inside2:right}
 \end{align}
With~\eqref{eq:traces},~\eqref{eq:xequationv}, and \eqref{inside2:right}, we arrive at the desired result \eqref{eq:fisherinfoineq}.
Finally, setting $\omega_q=\frac{J(\rho)^{-1}}{J(\rho)^{-1}+tJ(f)^{-1}}$ and $\omega_c=\frac{\sqrt{t}J(f)^{-1}}{J(\rho)^{-1}+tJ(f)^{-1}}$, we obtain \eqref{eq:Stam}.

\end{proof}

In the next lemma we show how the quantum Stam inequality implies an isoperimetric inequality for the quantum Fisher information.

\begin{lemma} (Quantum Fisher information isoperimetric inequality)
The following inequality holds:
\begin{align}\label{eq:iso}
  \frac{d}{dt}\bigg|_{t=0} \Bigl[\frac{1}{2d}J(e^{t\cLh}(\rho))\Bigr]^{-1} \geq 1\ .
\end{align}
\end{lemma}
\begin{proof}
Recall that we have $ e^{t\cLh}(\rho)=f_Z \star_t \rho$ for a Gaussian random variable $Z$ \eqref{eq:def:Gaussian}.
In the quantum Stam inequality \eqref{eq:Stam}, take $f=f_Z$, then
 \begin{align}
   \frac{1}{t}\left( J(f_Z\star_t \rho)^{-1}-J(\rho)^{-1}\right)\geq 
   J(f_Z)^{-1}=(2d)^{-1}\ .
 \end{align}
 Taking the limit $t\rightarrow 0$, we arrive at the desired inequality.
\end{proof}

The isoperimetric inequality is tight for one mode ($d = 1$) and saturated by the Gaussian thermal state
\begin{align}
  \omega_\n = \frac{1}{\n+1}\sum_{j=0}^\infty \left(\frac{\n}{\n+1}\right)^j\proj{j}
  \label{eq:gaussianthermal}
\end{align}
with mean-photon number $\n$: As shown in Appendix \ref{app:Fisher}, we have
\begin{align}
  \frac{d}{dt}\bigg|_{t=0}\Bigl[\frac{1}{2}J(e^{t\cLh}(\omega_\n))\Bigr]^{-1} &=\frac{1}{\n(\n+1)}\log^{-2}\left(1+\frac{1}{\n}\right)\rightarrow\ 1\qquad \text{ as }\n\rightarrow\infty\ .
\end{align}

As in classical information theory, the isoperimetric inequality for the quantum Fisher information implies concavity of the entropy power under diffusion as an immediate consequence.
We define the entropy power as
\begin{align}
N(\rho) = \exp\left(S(\rho)/d\right)\ .
  \label{def:entropy_power}
\end{align}

\begin{theorem} (Concavity of the quantum entropy power) 
  The entropy power along trajectories of the diffusion semigroup  \eqref{eq:def:diffusion} is concave, i.e.
 \begin{align}
   \frac{d^2}{dt^2}\bigg|_{t=0} N(e^{t\cLh}(\rho)) \leq 0\ .
 \end{align}
\end{theorem}
\begin{proof}
 Two applications of de Bruijn identity  \eqref{eq:de_Bruijn_heat} yield
 \begin{align}
   \frac{d^2}{dt^2}\bigg|_{t=0} N(e^{t\cLh}(\rho))=N(\rho)\left(\Bigl[\frac{1}{2d}J(\rho)\Bigr]^2+\frac{1}{2d}\frac{d}{dt}\bigg|_{t=0}J(e^{t\cLh}(\rho)) \right)\ . 
 \end{align}
 The quantum Fisher information isoperimetric inequality \eqref{eq:iso} is equivalent to
 \begin{align}
   \frac{1}{2d}J(\rho)^2 + \frac{d}{dt}\bigg|_{t=0}J(e^{t\cLh}(\rho)) \leq 0\ .
 \end{align}
This completes the proof.
\end{proof}

\newcommand*{\cLhcl}{\cC_{\text{c}}}
To proceed, we establish bounds on the asymptotic scaling of the entropy power for large times. The following lemma follows directly from \cite{Blachman65} and ~\cite[Corollary~III.4]{KoeSmiEPI} (see also~\cite{KoeSmiErratum}). 

\begin{lemma}[Asymptotic scaling of the entropy power under the heat flow]
\label{lem:scaling}
Let 
\begin{align}
\cLhcl&=\frac{1}{2}\sum_{j=1}^{2d}\frac{\partial^2}{\partial \vec{\xi}_j^2}
\label{eq:clheat}
\end{align}
be the generator of the classical heat diffusion semigroup on~$\mathbb{R}^{2d}$. In the limit $t\rightarrow\infty$, we have 
\begin{align}
\exp\left(H(e^{t\cLhcl}(f))/d\right) &=(2\pi e)t+O(1)\ ,\\
\exp\left(S(e^{t\cLh}(\rho))/d\right) &=(2\pi e)t+O(1)\ 
\end{align}
independent of the probability density function~$f$ on $\mathbb{R}^{2d}$ and the $d$-mode state~$\rho$, respectively. 
\end{lemma}

Having the same scaling for classical and quantum heat flows motivated the choice of constants in~\eqref{eq:def:Weyl_op} and~\eqref{eq:def:diffusion}. 

Note that if $X$ is a random variable with probability density function~$f$, then
$e^{t\cLhcl}(f)$ is the probability density function of the random variable
$X+\sqrt{t}Z$ obtained by adding a centered Gaussian random variable with unit variance, see Eq. \eqref{eq:def:Gaussian}.

\begin{lemma}
  \label{lem:compability}
We have 
\begin{align}
  e^{\xi\cLh}(f\star_t \rho)&= e^{\nu \cLhcl}(f) \star_t e^{\mu\cLh}(\rho)
\end{align}
whenever $\xi = \mu + t\nu$.
\end{lemma}
\begin{proof}
Observe that writing $\star$ for $\star_1$, we get
\begin{align}
  e^{\xi\cLh}(f_X\star_t \rho)&= f_Z\star_\xi (f_{\sqrt{t}X}\star\rho)\\
&=f_{\sqrt{\xi}Z}\star (f_{\sqrt{t}X}\star \rho)\\
&=f_{\sqrt{\xi}Z+\sqrt{t}X}\star \rho\ .
\end{align}
On the other hand,
\begin{align}
  e^{\nu \cC_c}(f) \star_t e^{\mu\cLh}(\rho) &= f_{X + \sqrt{\nu}Z_1} \star_t \left(e^{\mu\cLh}(\rho)\right)\\
  &= f_{\sqrt{t}(X + \sqrt{\nu}Z_1)} \star \left(f_{\sqrt{\mu}Z_2}\star \rho \right)\\
  &= f_{\sqrt{t}X + \sqrt{t\nu}Z_1 + \sqrt{\mu}Z_2} \star \rho\\
  &= f_{\sqrt{t}X + \sqrt{\mu + t\nu}Z} \star \rho\\
  &= f_{\sqrt{t}X + \sqrt{\xi}Z} \star \rho\ .
\end{align}
We have used properties \eqref{eq:conv_scaling} and \eqref{eq:conv_addition}. In the penultimate step we have used that for independent unit-variance centered Gaussian random variables $Z_1$ and $Z_2$, we have $a Z_1 + bZ_2$ = $\sqrt{a^2+b^2}Z$.
Hence the two expressions are equal and the statement follows.

\end{proof}

The next theorem presents the entropy power inequality for both the convolution operation \eqref{eq:cF} and the heat diffusion semigroup \eqref{eq:def:diffusion}.

\begin{theorem}\label{thm:entropy_power}(Entropy power inequality)
  For $t\geq 0$, the following inequality holds
\begin{align}
  N(f \star_t \rho)\geq N(\rho) + t N(f)\ .
\end{align}
In particular, choosing $f = f_Z$ as the distribution of a unit-variance centered Gaussian defined in Eq. \eqref{eq:def:Gaussian}, we have
\begin{align}\label{eq:entropy_power}
  N(e^{t\cLh}(\rho))\geq  N(\rho)+t\; 2\pi e\ .
\end{align}
\end{theorem}
\begin{proof}
  The proof is inspired by the proof of the entropy power inequality in \cite{KoeSmiEPI}, which itself is inspired by the proof for classical random variables by Blachman \cite{Blachman65}. Here we provide all necessary details modified to the present situation. For $\mu,\nu,\xi \geq 0$, define functions
 \begin{align}
  &\mu\rightarrow E_A(\mu):=\exp\left(S(e^{\mu\cLh}(\rho))/d\right)\ ,\\
  &\nu\rightarrow E_B(\nu):=\exp\left(H(e^{\nu\cL_{\text{heat,cl}}}(f))/d\right)\ ,\\
  &\xi\rightarrow E_C(\xi):=\exp\left(S(e^{\xi\cLh}(f \star_t \rho)))/d\right)\ ,
 \end{align}
 where $\cC_c$ is the generator of the classical heat semigroup as defined in Eq. \eqref{eq:clheat}.
 
The initial value problems
 \begin{align}
  &\dot{\mu}(s)=E_A(\mu(s))\ , \ \ \mu(0)=0\ ,\\
  &\dot{\nu}(s)=E_B(\nu(s))\ , \ \ \nu(0)=0\label{inside:IVP}
 \end{align}
 have solutions $\mu(\cdot), \nu(\cdot)$. Fix a pair of such solutions $(\mu(\cdot), \nu(\cdot))$ and $t\geq 0$. Define 
 \begin{align}\label{inside:H}
  \xi(s):=\mu(s)+t\nu(s)\ .
 \end{align}
 These functions diverge, i.e.
 \begin{align}\label{inside:div}
  \lim_{s\rightarrow\infty}\mu(s)=\lim_{s\rightarrow\infty}\nu(s)=\lim_{s\rightarrow\infty}\xi(s)=\infty\ ,
 \end{align}
 because of \eqref{inside:IVP} and $E_{A,B} \geq 1$.
 
 Consider the function
 \begin{align}
   \delta(s):=\frac{E_A(\mu(s))+ t E_B(\nu(s))}{E_C(\xi(s))}\ .
 \end{align}
 With the initial conditions \eqref{inside:IVP}, it follows that the claim of the theorem is equivalent to 
 \begin{align}
  \delta(0)\leq 1\ .
 \end{align}
 This inequality follows from two facts: first, the fact that
 \begin{align}
  \lim_{s\rightarrow\infty}\delta(s)=1\ ,
 \end{align}
as follows from the asymptotic scaling shown in Lemma \ref{lem:scaling},
the divergence \eqref{inside:div}, and the choice~\eqref{inside:H} of $\xi(s)$; 
second, the fact that
\begin{align}
  \dot{\delta}{(s)}\geq 0 \qquad \text{ for all }s\geq 0\ .
  \label{claimdelta}
\end{align}
It remains to show identity~\eqref{claimdelta}. Computing the derivative of $\delta$ leads to the following equality
\begin{align}\label{inside:delta_der}
  \dot{\delta}(s)&=\frac{\dot{E}_A(\mu(s))\dot{\mu}(s)+t\dot{E}_B(\nu(s))\dot{\nu}(s)}{E_C(\xi(s))}-\frac{E_A(\mu)+tE_B(\nu)}{E_C(\xi)^2}\dot{E}_C(\xi(s))\dot{\xi}(s)\ .
\end{align}
Define the Fisher informations
\begin{align}
  &J_A(\mu):=J(e^{\mu\cLh}(\rho))\ ,\\
 &J_B(\nu):=J(e^{\nu \cLhcl}(f))\ ,\label{inside:Fisher}\\
 &J_C(\xi):=J(e^{\xi\cLh}(f \star_t \rho))\ ,
\end{align}
for $\mu, \nu, \xi \geq 0$.
From the quantum de Bruijn identity \eqref{eq:de_Bruijn_heat} and the classical de Bruijn identity \eqref{eq:de_Bruijn_cl} we obtain
\begin{align}\label{inside:deBruijn}
 \dot{E}_V(\zeta)=\frac{1}{2d}E_V(\zeta)J_V(\zeta)\qquad \text{ where }V\in\{A,B,C\}\ .
\end{align}
With these identities and Eq. \eqref{inside:IVP}, Eq. \eqref{inside:delta_der} is equivalent to
\begin{align}
 2d\dot{\delta}(s)&=\frac{E_A^2J_A+tE_B^2J_B}{E_C}-\frac{(E_A+tE_B)^2}{E_C^2}E_CJ_C\ ,
 \label{intermediatedelta}
 \end{align}
where we have used the shorthand notation $E_A=E_A(\mu(s))$, $E_B=E_B(\nu(s))$, and $E_C=E_C(\xi(s))$, and similarly for $J_{A,B,C}$.

Recall that by Lemma \ref{lem:compability}, $e^{\xi \cLh}\left(f \star_t\rho\right) = e^{\nu\cLhcl}(f) \star_t e^{\mu \cLh}(\rho)$, and by the quantum Stam inequality \eqref{eq:Stam}, we have the bound
\begin{align}\label{inside:J_Zineq}
 J_C\leq \frac{J_AJ_B}{tJ_A+J_B}\ .
\end{align}
Inserting this upper bound into \eqref{intermediatedelta}, we obtain
\begin{align}
 2 d\dot{\delta}E_C&\geq E_A^2J_A+tE_B^2J_B-(E_A+tE_B)^2\frac{J_AJ_B}{tJ_A+J_B} = \frac{t(E_AJ_A-E_BJ_B)^2}{tJ_A+J_B}\geq 0\ .
\end{align}
This proves \eqref{claimdelta}.
\end{proof}

As the last statement in this section, we derive an isoperimetric inequality for entropies from the entropy power inequality.

\begin{theorem} (Isoperimetric inequality for entropies)
  We have
\begin{align}\label{eq:iso_entropy}
 \frac{1}{d}J(\rho)N(\rho)\geq 4\pi e\ .
\end{align} 
\end{theorem}
\begin{proof}
  Applying the de Bruijn identity \eqref{eq:de_Bruijn_heat} to the definition \eqref{def:entropy_power} of the entropy power $N(\rho)$, we obtain
\begin{align}
  \frac{d}{dt}\bigg|_{t=0} N(e^{t\cLh}(\rho))=\frac{1}{2d}J(\rho)N(\rho)\ .
\end{align}
On the other hand, for $t\geq 0$, the entropy power inequality \eqref{eq:entropy_power} reduces to
\begin{align}
  \frac{1}{t}[N(e^{t\cLh}(\rho))-N(\rho)]\geq 2\pi e\ .
\end{align}
Therefore, taking the limit $t\rightarrow 0$, we obtain the desired bound.
\end{proof}

To conclude this section, we remark that the isoperimetric inequality for entropies \eqref{eq:iso_entropy} is tight in the one-mode case, $d=1$:
For a Gaussian thermal state $\omega_\n$ \eqref{eq:gaussianthermal} with mean photon number $\n$ we obtain
\begin{align}
  J(\omega_\n)N(\omega_\n) = 4\pi \left(\frac{\n+1}{\n}\right)^{\n}\log\left(\frac{\n+1}{\n}\right)^{\n+1} \rightarrow 4\pi e\qquad \text{ for } \n \rightarrow \infty\ .
\end{align}
Detailed calculations are provided in Appendix \ref{app:isoperim}.

\section{Gaussian optimality for energy-constrained entropy rates}\label{sec:gaussian}

In this section we show that Gaussian thermal states minimize the entropy rate for the one-mode attenuator semigroup among states with bounded mean photon number.  The semigroup is defined by its generator
\begin{align}\label{eq:L_-}
  \cL_-(\rho)= a\rho a^\dagger-\frac{1}{2}\{a^\dagger a, \rho\}\ . 
\end{align}
We prove the following theorem:
\begin{theorem}\label{thm:gaussian_optimal}
For any $\n > 0$, the infimum
$\inf_{\rho:\Tr(\hat{n}\rho)\leq\n} \frac{d}{dt}\big|_{t=0} S(e^{t\cL_-}(\rho))$ over states $\rho$ with mean photon number $\n$
is achieved by the  Gaussian thermal state~$\omega_\n$ defined in \eqref{eq:gaussianthermal}.
In particular,
\begin{align}
  \inf_{\rho:\Tr(\hat{n}\rho)\leq\n} \frac{d}{dt}\Big|_{t=0} S(e^{t\cL_-}(\rho))= \begin{cases} -\n\log\left(1+\frac{1}{\n}\right) &\textrm{ if } \n > 0\ ,\\ 0 & \textrm{ if } \n = 0\ .\end{cases}
\end{align} 
\end{theorem}
The proof proceeds by reduction to a recent result by De Palma, Trevisan, and Giovannetti~\cite{depalma2016}, where it is shown that Gaussian thermal states minimize the entropy rate of the quantum attenuator among all states with a given input entropy\footnote{More precisely,~\cite[Theorem~6]{depalma2016} states that the entropy rate of a state with finite support is lower bounded by that of a Gaussian state with the same entropy. Since we are interested in a statement applicable to arbitrary states, we use their Theorem~24 instead: the latter gives an analogous statement for the corresponding classical process but without a finiteness assumption.}. Our argument additionally uses the recently introduced concept of Fock majorization and associated results by Jabbour, Garcia-Patron and Cerf~\cite{jabbouretal}, as well as the (classical) Gaussian maximum entropy principle.

In more detail, we first show that the problem of minimizing the entropy rate reduces to  the study of properties of a classical semigroup describing a pure-death process. This connection was used previously in~\cite{carbonesasso07}  (see Section~\ref{sec:appornstein}) and is also an essential first step in~\cite{depalma2016}. More explicitly, in Section~\ref{sec:gaussian:conn_class} (Theorem~\ref{thm:correspondence_class}) we prove the  identity
\begin{align}
\inf_{\rho:\Tr(\hat{n}\rho)\leq\n} \frac{d}{dt}\Big|_{t=0} S(e^{t\cL_-}(\rho))
&=
\inf_{p:\mathbb{E}_p[N]\leq\n} \frac{d}{dt}\Big|_{t=0} H(e^{t\cC_-}(p))\ .
\end{align}
Here 
the infimum is over all probability distributions $p$ on $\mathbb{N}_0$ with expectation value $\mathbb{E}_p[N]$ bounded   by~$\n$,
the quantity~$H(p)$ is  the Shannon entropy of the distribution~$p$, and $\cC_-$ is the generator of a semigroup describing a classical pure-death process (see~\eqref{eq:puredeath} for a precise definition). 

Using the results of~\cite{depalma2016}, we then show that the entropy rate for the classical process is optimized by a geometric distribution: In Section~\ref{sec:gaussian:class} (Theorem~\ref{thm:correspondence_class}) we prove that for $\n > 0$
\begin{align}
\inf_{p:\mathbb{E}_p[N]\leq\n} \frac{d}{dt}\Big|_{t=0} H(e^{t\cC_-}(p))=-\n\log\left(1+\frac{1}{\n}\right)\ .
\end{align}
Finally, in Section \ref{sec:gaussian:gaussian}, we calculate the entropy rate for the Gaussian thermal state $\omega_\n$ with mean photon number at most $\n$ and find that 
\begin{align}
\frac{d}{dt}\Big|_{t=0} S(e^{t\cL_-}(\omega_\n))=-\n\log\left(1+\frac{1}{\n}\right).
\end{align}

\subsection{Connection to a classical pure-death process}\label{sec:gaussian:conn_class}
For an  initial  state $\rho=\sum_n p_n \proj{n}$
which is diagonal in the number state basis~$\{\ket{n} = \frac{\left(a^\dagger\right)^n}{\sqrt{n!}} \ket{0}\}_{n\in\mathbb{N}_0}$  (where $\hat{n}\ket{0}=0$), 
the time-evolved state has the same form, i.e.,  $\rho(t)=e^{t\cL_{-}}(\rho)=\sum_n p_n(t) \proj{n}$.  Thus
the attenuator semigroup $\{e^{t\cL_-}\}_{t\geq 0}$ 
gives rise to a semigroup~$\{e^{t\cC_-}\}_{t\geq 0}$  on classical probability distributions by $p(t)=e^{t\cC_-}(p)$. Its generator~$\cC_-$ describes the dynamics 
 of a classical pure-death process. It can be obtained from \eqref{eq:L_-} by inserting a number state $\ket{n}$: it is straightforward to check that 
\begin{align}
  \cL_-(\proj{n})&=\begin{cases} n(\proj{n-1}-\proj{n}) & \text{ for } n > 0\ ,\\ 0 & \text{ for } n = 0\ .\end{cases}
\end{align}
In particular,  the coefficients $\{p_n(t)\}_{n\in\mathbb{N}_0}$ satisfy the system of differential equations
\begin{align}
  \dot{p}_n(t) &=-np_n(t)+(n+1)p_{n+1}(t) \qquad \text{ for all } n \in \mathbb{N}_0\label{eq:puredeath}\ ,
\end{align}
with initial condition $p_n(0)=p_n$ for $n\in\mathbb{N}_0$. The expression on the right-hand side~of~\eqref{eq:puredeath} defines the generator~$\cC_-$, that is, we have 
\begin{align}
   \left(\cC_-(p)\right)_n = - n p_n + (n+1) p_{n+1} \;\;\; \text{ for all } n \in \mathbb{N}_0\ .
  \label{eq:c-action}
\end{align}

The following theorem reduces the problem of minimizing the entropy rate for the quantum attenuator semigroup to the classical problem of minimizing the entropy rate for this pure-death process. 
\begin{theorem}[Correspondence to classical problem]\label{thm:correspondence_class}
We have the identity
\begin{align}
  \inf_{\rho:\Tr(\hat{n}\rho)\leq\n} \frac{d}{dt}\Big|_{t=0} S(e^{t\cL_-}(\rho))
&=
\inf_{p:\mathbb{E}_p[N]\leq\n} \frac{d}{dt}\Big|_{t=0} H(e^{t\cC_-}(p))\ ,
\end{align}
where $\mathbb{E}_p[N]=\sum_{n=0}^\infty n p_n$ 
and $H(p)=- \sum_{n=0}^\infty p_n \log p_n$  are the expectation value and entropy of the  distribution~$p$, respectively. 
\end{theorem}

The proof of Theorem~\ref{thm:correspondence_class} relies on results obtained in~\cite{depalmaetal} and~\cite{jabbouretal}. Here we review the necessary definitions and results, and specialize them  to our situation.

\begin{definition}[Majorization] Let $p$ and $q$ be decreasing summable sequences of positive numbers. Then $p$ weakly sub-majorizes $q$, $q\prec_w p$, if and only if $$\sum_{i=0}^n q_i\leq \sum_{i=0}^n p_i \qquad \text{ for all }n\in\mathbb{N}.$$
\end{definition}

\begin{definition}
Let $P$ and $Q$ be positive trace-class operators with eigenvalues $\{p_n\}_{n\in\mathbb{N}}$ and $\{q_n\}_{n\in\mathbb{N}}$, arranged in decreasing order. Then $P$ weakly majorizes $Q$, i.e., $Q\prec_w P$, if and only if $q\prec_w p$. Also, $P$ majorizes $Q$, i.e., $Q\prec P$, if the traces of $P$ and $Q$ are identical and $Q \prec_w P$.
\end{definition}

\begin{definition}[Fock rearrangement] Let $X$ be a positive trace-class operator with eigenvalues $\{x_n\}_{n\in\mathbb{N}_0}$ in decreasing order. The Fock rearrangement (or passive rearrangement) is defined as
$$X^\downarrow:=\sum_{n=0}^\infty x_n\ket{n}\bra{n}. $$
\end{definition}

Our restriction on the mean photon number forces us to consider an additional majorization relation. In contrast, in~\cite{depalma2016}, where the authors restrict the input entropy instead, Fock majorization does not need to be considered. 
 
 \begin{definition}[Fock majorization]
 The Fock majorization relation, denoted~$\prec_F$, was introduced in~\cite{jabbouretal} as follows:
 \begin{align}
 \sigma \prec_F \rho\qquad \Leftrightarrow \qquad \Tr(\Pi_n\sigma)\leq \Tr(\Pi_n\rho)\quad\textrm{ for all }n\in\mathbb{N}_0\ ,\label{eq:fockmajorizationrelatdef}
 \end{align}
 where \begin{align}
   \label{eq:projfock}
   \Pi_n=\sum_{j=0}^n \proj{j}.\end{align}
\end{definition} 

\begin{theorem}(\cite{depalmaetal})\label{thm:entropy_ineq_Fock}
For any state $\rho$ and all $t\geq 0$ we have
 \begin{align}
 e^{t\cL_-}(\rho)\prec e^{t\cL_-}(\rho^\downarrow)\ .\label{eq:majorizationprec}
 \end{align}
 This implies that 
\begin{align}
 S\left(e^{t\cL_-}(\rho)\right)\geq S\left(e^{t\cL_-}(\rho^\downarrow)\right)\ .\label{eq:entropymajorizationdepalma}
 \end{align}
\end{theorem}
\begin{proof}
The first statement is shown in~\cite[Eq. (VI.10)]{depalmaetal}, and the second statement then follows from \cite[Theorem III.3]{depalmaetal}.
\end{proof}
 
  \begin{lemma}(\cite[Lemma IV.9]{depalmaetal})\label{lem:depalmamajorx}
 Suppose $X,Y,Z$ are positive trace-class operators with
 \begin{align}
 Y\prec_w Z\qquad\textrm{ and }\qquad Z^\downarrow= Z\ .
 \end{align}
 Then
 \begin{align}
 \Tr(XY)\leq \Tr(X^\downarrow Z)\ .
 \end{align}
 \end{lemma}
 
 We begin proving Theorem \ref{thm:correspondence_class} by investigating the change of the mean photon number under the Fock rearrangement procedure.
 
 \begin{lemma}\label{lem:photonnumberrearrange}
 For any state~$\rho$, the Fock rearrangement does not increase the mean photon number
 \begin{align}
 \Tr(\hat{n}\rho^\downarrow)\leq \Tr(\hat{n}\rho)\ ,
 \end{align}
 where $\hat{n} = a^\dagger a$.
 \end{lemma}
 \begin{proof}
   Let $n\in\mathbb{N}_0$ be arbitrary, and set $X=\Pi_n$ with $\Pi_n$ defined in \eqref{eq:projfock}. Then clearly $X^\downarrow =X=\Pi_n$.
 Setting~$Y=\rho$, $Z=\rho^\downarrow$, we have $Z^\downarrow=Z$, and $Y\prec Z$, according to \eqref{eq:majorizationprec} for $t=0$, i.e.
 \begin{align}
 \rho\prec \rho^\downarrow\ .\label{eq:preceqrhorhofock}
 \end{align}
  Thus, Lemma~\ref{lem:depalmamajorx} leads to
$
 \Tr(\Pi_n\rho)\leq \Tr(\Pi_n\rho^\downarrow).
$
 According to Definition~\eqref{eq:fockmajorizationrelatdef}, this is equivalent to $$\rho\prec_F\rho^\downarrow.$$ It was shown in (\cite[Eq.~(5)]{jabbouretal}) that
 \begin{align}
 \rho\prec_F\sigma \qquad\Rightarrow\qquad \Tr(\hat{n}\sigma)\leq \Tr(\hat{n}\rho)\ .\label{eq:fockmajorizationv}
 \end{align}  
 The claim follows by taking $\sigma=\rho^\downarrow$ in the last statement.
 \end{proof}

\begin{proof}[Proof of Theorem \ref{thm:correspondence_class}]
Since $S(\rho^\downarrow)=S(\rho)$, Theorem \ref{thm:entropy_ineq_Fock} implies
\begin{align}
\frac{d}{dt}\Big|_{t=0}S\left(e^{t\cL_{-}}(\rho)\right)\geq \frac{d}{dt}\Big|_{t=0}S\left(e^{t\cL_{-}}(\rho^\downarrow)\right)\ .
\label{eq:majresult}
\end{align}
With Lemma~\ref{lem:photonnumberrearrange}, we therefore obtain
\begin{align}\label{eq:relation_class}
\inf_{\rho:\Tr(\hat{n}\rho)\leq \n}\frac{d}{dt}\Big|_{t=0}S\left(e^{t\cL_{-}}(\rho)\right) = \inf_{\rho = \rho^\downarrow:\Tr(\hat{n}\rho^\downarrow)\leq\n}\frac{d}{dt}\Big|_{t=0}S\left(e^{t\cL_{-}}(\rho^\downarrow)\right)\ .
\end{align} 
Since $\rho^\downarrow$ is a passive, Fock-rearranged state, the right-hand side is a classical problem related to the pure-death process~\eqref{eq:puredeath}, and, therefore, can be replaced with the infimum of the entropy rate of a probability distribution evolving under a pure-death process.  Thus the claim follows.
\end{proof}

\subsection{Geometric distributions optimize entropy rates of the classical death process under energy constraint}\label{sec:gaussian:class}
\newcommand{\cG}{\mathcal G}
For a probability distribution $p$ on~$\mathbb{N}_0$, let 
\begin{align}
J_-(p):=2\frac{d}{dt}\Big|_{t=0}H(e^{t\cC_-}(p))
\end{align}
denote the entropy rate when~$p$ evolves under the classical death-process~$\cC_-$ (cf.~\eqref{eq:c-action}). We are interested in distributions~$p$ with a fixed expectation value~$\mathbb{E}_p[N]=\sum_{n=0}^\infty n p_n$. 
The main
result we use here is~\cite[Theorem 24]{depalma2016}: it states that
for any probability distribution~$p$ on~$\mathbb{N}_0$, the quantity $J_-(p)$ is bounded by
\begin{align}
 \inf_{p: H(p) \leq H} J_-(p)\geq 2\inf_{p:H(p)\leq g(\n) } f(H(p))\ \label{eq:depalmaeq}
\end{align}
where $f(H) = - g^{-1}(H) g'(g^{-1}(H))$ 
and where  $g(\n)=(\n+1)\log(\n+1) - \n\log\n$ is the entropy of 
a geometric distribution with expectation value~$\n$ (or equivalently the entropy of a Gaussian state with mean photon number~$\n$). We use this to show the following:
\begin{theorem}
  \label{thm:alldistributionslowerbound}
The infimum $  \inf_{p: \mathbb{E}[N] \leq  \n} J_-(p)$ is achieved by the geometric distribution 
$p_k^{\mathrm{geo},\n} = (1-r)r^k$ with $r = \frac{\n}{\n+1}$.
In particular, for $\n > 0$,
\begin{align}
  \inf_{p: \mathbb{E}[N] \leq  \n} J_-(p) = - 2\n \log\left(1 + \frac{1}{\n}\right)\ .
\end{align}
\end{theorem}

\begin{proof}
We show that 
\begin{align}
  \inf_{p: \mathbb{E}[N] =  \n} J_-(p) =J(p^{\mathrm{geo},\n})= - 2\n \log\left(1 + \frac{1}{\n}\right)\ .\label{eq:toshownvf}
\end{align}
Since the right-hand side is monotonically decreasing with $\n$, Eq.~\eqref{eq:toshownvf} implies the claim of the theorem.

The geometric distribution $p^{\mathrm{geo},\n}~=~(1-r)r^k$ with $r=\frac{\n}{\n+1}$ has expectation value $\mathbb{E}_{p^{\mathrm{geo},\n}}[N] = \n$ and entropy~$H(p^{\mathrm{geo},\n})=g(\n)$. By the maximum entropy principle \cite[Chapter 12]{coverelements}, we know that geometric distributions are the distributions with maximal entropy among all distributions with a fixed expectation value $\mathbb{E}_p[N]$.
Therefore we have 
\begin{align}
  \mathbb{E}_p[N] = \n \qquad \Rightarrow\qquad H(p) \leq g(\n)\ .
\end{align}
Combining this with~\eqref{eq:depalmaeq} we  obtain
\begin{align}
  \inf_{p: \mathbb{E}_p[N] = \n} J_-(p) \geq \inf_{p: H(p) \leq g(\n)} J_-(p) \geq 2\inf_{p: H(p) \leq g(\n)}  f(H(p))\ ,
\end{align}
 Since $f$ is decreasing \cite[Lemma 5]{depalma2016}, it follows that
\begin{align}
  \inf_{p:\mathbb{E}[N] = \n} J_-(p) \geq 2 f(g(\n))= -2 \n \log\left(1 + \frac{1}{\n}\right)\ .
  \label{eq:jminusboundcl}
\end{align}
However, since 
$\mathbb{E}_{p^{\mathrm{geo},\n}}[N] = \n$ and $J_-(p^{\mathrm{geo},\n}) = -2 \n \log\left(1 + \frac{1}{\n}\right)$, we have equality in \eqref{eq:jminusboundcl}.
\end{proof}
\subsection{Gaussian optimality of entropy rates for the quantum attenuator semigroup}\label{sec:gaussian:gaussian}

\begin{proof}[Proof of Theorem \ref{thm:gaussian_optimal}]

From Theorem \ref{thm:correspondence_class} and Theorem \ref{thm:alldistributionslowerbound} we have 
\begin{align}
  \inf_{\rho:\Tr(\hat{n}\rho)\leq\n} \frac{d}{dt}\big|_{t=0} S(e^{t\cL_-}(\rho))=-\n\log\left(1+\frac{1}{\n}\right).
\end{align}
Let $\omega_\n$ be the Gaussian thermal state with mean photon number $\n$ as defined in \eqref{eq:gaussianthermal}. We have 
$
S(\omega_\n)=g(\n)=(\n+1)\log (\n+1)-\n\log \n$ for $\n>0.
$
Under the map $e^{t\cL_-}$, the state $\omega_\n$ evolves into the thermal state $\omega_{\n_t}$ according to
\begin{align*}
e^{t\cL_-}(\omega_\n)&=\omega_{\n_t}\ ,
\qquad\textrm{ where }\qquad \n_t=e^{-t}\n\ .
\end{align*}
In particular,
\begin{align*}
\frac{d}{dt}\Big|_{t=0}S(e^{t\cL_-}(\omega_\n))&=
g'(\n)\n'_t\Big|_{t=0}=-\n\log\left(1+\frac{1}{\n}\right)\ .
\end{align*}
Therefore the considered infimum is achieved by the Gaussian thermal state $\omega_\n$.

\end{proof}

\section{Application to fast convergence of the Ornstein-Uhlenbeck semigroup}\label{sec:app}

In this section we consider a one-parameter group of CPTP maps $\{e^{\cL_{\mu,\lambda}}\}_{t\geq 0}$ generated by the linear combination
\begin{align*}
  \cL_{\mu,\lambda}&=\mu^2\cL_-+\lambda^2\cL_+ \qquad \textrm { for } \mu > \lambda > 0\ .
\end{align*}
where $\cL_-$ is defined by \eqref{eq:L_-}, and $\cL_+$ is defined by
\begin{equation}\label{eq:L_+}
 \cL_+(\rho)= a^\dagger\rho a-\frac{1}{2}\{a a^\dagger, \rho\}\ . 
 \end{equation}
In the following,
we use the entropy production rates 
\begin{align}
J_{\pm}(\rho)&:=2 \frac{d}{dt} S(e^{t\cL_{\pm}}(\rho))\ .
\end{align}
The factor $2$ here is for convenience to match de Bruijn identity \eqref{eq:de_Bruijn_heat} for $J(\rho)$. 
These quantities are related to the Fisher information $J(\rho)$ by
\begin{align}
J(\rho)&=2\pi (J_-(\rho)+J_+(\rho))\ ,\label{eq:sumidentityplusminus}
\end{align}
because
of de Bruijn's identity
and the fact that $\cLh=2\pi\cL_-+2\pi\cL_+$.

\begin{lemma}\label{lem:rate_decay}
Let $\mu>\lambda >0$. Then
\begin{align}
-\zeta D(\rho\|\sigma_{\mu,\lambda})-\frac{d}{dt}\Big|_{t=0} D(e^{t\cL_{\mu,\lambda}}(\rho)\|\sigma_{\mu,\lambda})= 
~&\frac{\mu^2}{2}J_-(\rho)+\frac{\lambda^2}{2}J_+(\rho)+\zeta S(\rho)\\
&+\lambda^2\log \nu+\zeta \log(1-\nu)
\label{eq:rate_decay}
\end{align} 
for any state $\rho$, where $\nu=\frac{\lambda^2}{\mu^2}$, $\zeta=\mu^2-\lambda^2$, and $\sigma_{\mu,\lambda} $ is the fixed point of $\cL_{\mu,\lambda}$.
\end{lemma}

\begin{proof}
  As mentioned in Section \ref{sec:appornstein}, the unique fixed point of the semigroup $\{e^{t\cL_{\mu,\lambda}}\}_{t\geq 0}$ is the state $$\sigma_{\mu,\lambda}=(1-\nu ) \sum_{n=0}^\infty \nu^n \proj{n}=(1-\nu)\nu^{\hat{n}},$$ 
where $\nu=\lambda^2/\mu^2$. In particular, this implies that for any state~$\rho$
\begin{align}
D(\rho\|\sigma_{\mu,\lambda})&=-S(\rho)-\Tr(\rho\log\sigma_{\mu,\lambda})
=-S(\rho)-(\log \nu)\Tr(\rho \hat{n})-\log(1-\nu)\ .\label{eq:drhosigmaexpl}
\end{align}
Straightforward calculations show that the mean photon number of $\rho$ converges to the mean photon number of the fixed point $\sigma_{\mu,\lambda}$ with an exponential rate 
\begin{align}
\n_t=\tr(e^{t\cL_{\mu,\lambda}}(\rho)\hat{n})&=\tr(\rho e^{t\cL^\dagger_{\mu,\lambda}}(\hat{n}))=e^{-(\mu^2-\lambda^2)t}\Tr(\rho \hat{n})+(1-e^{-(\mu^2-\lambda^2)t})\n_\infty\ ,
\end{align}
where $\n_\infty=\tr(\sigma_{\mu,\lambda}\hat{n})=\frac{\lambda^2}{\mu^2-\lambda^2}=\frac{\nu}{1-\nu}$. Therefore
\begin{align}\label{eq:n_der}
\frac{d}{dt}\Big|_{t=0}\Tr(e^{t\cL_{\mu,\lambda}}(\rho)\hat{n})=-(\mu^2-\lambda^2)\left(\Tr(\rho \hat{n})-\n_\infty\right).
\end{align}\label{eq:photon_number_convergence}
Combining the last equality with~\eqref{eq:drhosigmaexpl}, we find that
\begin{align}
-\frac{d}{dt}\Big|_{t=0} D(e^{t\cL_{\mu,\lambda}}(\rho)\|\sigma_{\mu,\lambda})&=\frac{\mu^2}{2}J_-(\rho)+\frac{\lambda^2}{2}J_+(\rho)-(\log \nu) (\mu^2-\lambda^2)\Tr(\rho \hat{n})+\lambda^2\log \nu\ .\label{eq:dtrhosigmacomp}
\end{align} 
With~\eqref{eq:drhosigmaexpl} and~\eqref{eq:dtrhosigmacomp},
and setting $\zeta=\mu^2-\lambda^2$,
we obtain the desired equality. 
\end{proof}

The choice of $\zeta=\mu^2-\lambda^2$ in the lemma is motivated by Gaussian states: in Appendix \ref{app:qOU} we show that for any Gaussian state $\rho$ the following inequality holds
\begin{align}
\frac{d}{dt}\Big|_{t=0} D(e^{t\cL_{\mu,\lambda}}(\rho)\|\sigma_{\mu,\lambda})\leq -\zeta D(\rho||\sigma_{\mu,\lambda})\qquad&
\text{ with }\qquad \zeta = \mu^2 - \lambda^2 > 0\ .
\end{align}
Furthermore, for any $\epsilon>0$ there exists a Gaussian state $\rho$ such that 
\begin{equation}
\frac{d}{dt}\Big|_{t=0} D(e^{t\cL_{\mu,\lambda}}(\rho)\|\sigma_{\mu,\lambda})\geq -(\zeta+\epsilon) D(\rho||\sigma_{\mu,\lambda})\ .
\end{equation}

Let us now consider a specific example of a quantum Ornstein-Uhlenbeck process.

\begin{example}\label{ex:2-1_process_S} Consider $\mu^2=2$, $\lambda^2=1$. Then
\begin{align}
\frac{d}{dt}\Big|_{t=0} D(e^{t\cL_{\sqrt{2},1}}(\rho)\|\sigma_{\sqrt{2},1})&\leq -D(\rho\|\sigma_{\sqrt{2},1})
\end{align}
for any state $\rho$ with $S(\rho) \gtrsim 2.06$. In comparison, the entropy of the fixed point is $S(\sigma_{\sqrt{2},1}) = 2\log(2) \approx 1.39$.
\end{example}
In Lemma \ref{lem:rate_decay} we can bound $J_+(\rho)\geq 2$~\cite[Eq. (43)]{buscemiwilde}, and the linear combination of $J_-(\rho)$ and $S(\rho)$ can be bounded by the result of de Palma et al. \cite{depalma2016}.
That is, for any state $\rho$ with $S(\rho)\geq S_0$,
\begin{align}
  \frac{\mu^2}{2} J_-(\rho)+\zeta S(\rho)&\geq \inf_{S\geq S_0}\left(\mu^2 f(S)+\zeta S\right)\ ,
\end{align}
where $f(S)=-g^{-1}(S)g'(g^{-1}(S))$, and $g(\n)=(\n+1)\log(\n+1)-\n\log(\n)$. Substituting
 $S=g(\n)$
gives
\begin{align} 
  \frac{\mu^2}{2} J_-(\rho)+\zeta S(\rho)&\geq \inf_{\n\geq g^{-1}(S_0)}\left(\mu^2 (-\n g'(\n))+\zeta g(\n)\right)=:F(S_0)\ .
\end{align}
That is, for any state with $S(\rho)\geq S_0$, we have
\begin{align}
-\zeta D(\rho\|\sigma_{\mu,\lambda})-\frac{d}{dt}\Big|_{t=0} D(e^{t\cL_{\mu,\lambda}}(\rho)\|\sigma_{\mu,\lambda})&\geq F(S_0)+\lambda^2+\lambda^2\log \nu+\zeta \log(1-\nu)\ .
\end{align}
With the choice $\mu^2=2$, $\lambda^2=1$, the function $F(S)$ is monotonically increasing for $S\gtrsim 0.5$. For $S_0\gtrsim 2.06$, the right-hand side of the last inequality is non-negative.
\boxendproof

The isoperimetric inequality for entropies \eqref{eq:iso_entropy} for one mode ($d = 1$) can be written as
 \begin{equation*}
  -S(\rho)\leq  \log \left(\frac{1}{4\pi e }J(\rho) \right)\ .
 \end{equation*}
 For any $A>0$, using $\log x\leq x-1$, we get
\begin{align}
\log \Bigl(\frac{1}{4\pi e }J(\rho) \Bigr)&=\log \Bigl(\frac{A J(\rho) }{4\pi e A } \Bigr)\\
&=\log \Bigl(\frac{1}{4\pi e A}\Bigr)+\log \Bigl(A J(\rho) \Bigr)\\
&\leq AJ(\rho)-2-\log (4\pi A).
\end{align} 
Therefore
  \begin{align}\label{eq:entropy_iso}
    -S(\rho)\leq AJ(\rho)-(2+\log (4\pi A))\ \text{ for } A > 0\ .
  \end{align}

\begin{lemma}[Log-Sobolev inequality for the qOU semigroup]\label{lem:logsobolev}
Let $\mu>\lambda >0$ and $\zeta>0$. Then
\begin{align}
-\zeta D(\rho\|\sigma_{\mu,\lambda})-\frac{d}{dt}\Big|_{t=0} D(e^{t\cL_{\mu,\lambda}}(\rho)\|\sigma_{\mu,\lambda}) &\geq \alpha_-J_-(\rho)+\alpha_+J_+(\rho)+\gamma\Tr(\rho \hat{n})+\delta
\label{eq:logsobolevstatement}
\end{align}
for all states $\rho$, where 
\begin{align}
\alpha_- &=\mu^2/2-2\pi A\zeta\\
\alpha_+ &=\lambda^2/2-2\pi A\zeta\\
\gamma &=(\log \nu)\left(\zeta-(\mu^2-\lambda^2)\right) \\
\delta &=\zeta \left(\log(1-\nu)+2+\log (4\pi A)\right)+\lambda^2\log \nu
\end{align}
for any $A>0$. Here  $\nu=\frac{\lambda^2}{\mu^2}$, and $\sigma_{\mu,\lambda} $ is the fixed point of $\cL_{\mu,\lambda}$. In particular, choosing $\zeta = \mu^2 - \lambda^2$ and $A = \frac{\lambda^2}{4\pi\left(\mu^2 - \lambda^2\right)}$, we have
\begin{align}
  -\zeta D(\rho\|\sigma_{\mu,\lambda})-\frac{d}{dt}\Big|_{t=0} D(e^{t\cL_{\mu,\lambda}}(\rho)\|\sigma_{\mu,\lambda}) &\geq -\zeta \n \log(1+1/\n)+\delta\ ,
  \label{eq:logsobolevex}
\end{align}
where $\n = \Tr(\rho \hat{n})$.
\end{lemma}

\begin{proof}
Combining \eqref{eq:sumidentityplusminus} and \eqref{eq:entropy_iso} leads to
\begin{align}
  S(\rho)\geq -2\pi AJ_-(\rho)-2\pi A J_+(\rho)+2+\log (4\pi A)\ .
  \end{align}
  Using this inequality in Lemma \ref{lem:rate_decay} we obtain the result \eqref{eq:logsobolevstatement}.

For the second part of the statement, Theorem \ref{thm:gaussian_optimal} shows that 
\begin{align}
J_-(\rho)\geq J_-(\omega_{\n})=-2\n \log(1+1/\n)\ ,\label{eq:jminuslowerbound}
\end{align}
where $\omega_\n$ is the Gaussian thermal state with mean photon number~$\n$. Eq.~\eqref{eq:logsobolevex} 
then follows directly from the first part of the statement with the choice $A=\frac{\lambda^2}{4\pi(\mu^2-\lambda^2)}$, inequality~\eqref{eq:jminuslowerbound}, and the choice of $\zeta=\mu^2-\lambda^2$.
\end{proof}

Considering the same specific qOU process as in Example \ref{ex:2-1_process_S} we obtain the same rate of convergence to its fixed point, but now only for states with low mean photon number instead of large initial entropy as in Example \ref{ex:2-1_process_S}.

\begin{example}\label{ex:2-1_process_n} Consider $\mu^2=2$ and $\lambda^2=1$. Then 
\begin{align}
\frac{d}{dt}\Big|_{t=0} D(e^{t\cL_{\sqrt{2},1}}(\rho)\|\sigma_{\sqrt{2},1})&\leq -D(\rho\|\sigma_{\sqrt{2},1})
\end{align}
for any state $\rho$ with $\tr(\rho \hat{n}) \lesssim 0.67$. In comparison, the mean photon number of the fixed point is $\Tr(\sigma_{\sqrt{2},1}\hat{n}) = 1$.
\end{example}
Choosing $\mu^2=2$ and $\lambda^2=1$ in Lemma \ref{lem:logsobolev}, we obtain
\begin{align}
-\zeta D(\rho\|\sigma_{\mu,\lambda})-\frac{d}{dt}\Big|_{t=0} D(e^{t\cL_{\mu,\lambda}}(\rho)\|\sigma_{\mu,\lambda}) &\geq - \n \log(1+1/\n)+2-2\log(2)\ .
\end{align}
The right-hand side is monotonically decreasing and for $\n\lesssim 0.67$, it is non-negative.
\boxendproof

\section{Discussion}
We have established new information-theoretic inequalities for bosonic quantum systems. Our inequalities are motivated by and directly generalize well-known existing results concerning the sum of two real-valued random variables.  They also complement recent results concerning the ``addition'' of two bosonic quantum states by means of a beamsplitter: we consider a hybrid operation which amounts to a certain way of combining a classical probability distribution on phase space with a quantum state. Mathematically, our work thus makes progress towards a unified view of three types of convolution operations: the convolution of two classical probability density functions, of two quantum Wigner functions, and of a pair consisting of a classical probability density function and a quantum Wigner function. Operationally, our results extend entropic characterizations of classical additive noise  and quantum additive noise to  so-called classical noise in bosonic systems. Indeed, the proofs of our main inequalities, which include  hybrid versions of the Fisher information and entropy power inequalities,  are formally very similar to existing proofs in the fully classical as well as fully quantum settings.

The consideration of the hybrid classical-quantum convolution operation~\eqref{eq:cqconvolution} brings additional advantages, however: it allows for the study of infinitesimal Gaussian perturbations to a given quantum state. In contrast, existing fully quantum entropy power inequalities are not directly amenable to such arguments (at least not in an obvious way) since basic uncertainty relations prevent us from making  sense of e.g., a Gaussian state with infinitesimal variance. Mirroring the derivation of the isoperimetric inequality from the Brunn-Minkowski inequality (where a given set is perturbed by adding an infinitesimally small ball), we obtain a quantum isoperimetric inequality relating Fisher information and entropy power. A striking simple consequence of the latter is the statement that the entropy power is a concave function of time for the quantum heat diffusion semigroup: again, this provides a quantum generalization of a fundamental result about  the classical heat equation.

Let us conclude by mentioning a few potential directions for future work, as well as some open problems. For concreteness and simplicity, we have considered the simplest non-trivial definition of a hybrid convolution operation defined on multiple modes, as studied in~\cite{WernerHarmonicanalysis84}. One may  generalize our convolution~\eqref{eq:cqconvolution} and the associated results by considering additional (linear) operations along the lines of~\cite{MariPalma15}. From the point of view of information theory, it is also interesting to examine the implications of our results for the capacity of the  classical noise channels similarly to~\cite{KoeSmiEPIChannel}.

On a more speculative side, one may wonder whether alternative quantum generalizations of the results considered here exist, especially related to the conjectured photon number inequality by Guha, Erkmen and Shapiro~\cite{Guhaetal07}. These authors (and subsequent work such as~\cite{guha16}) suggest replacing
the quantum entropy power~$e^{S(\rho)}$ by the arguably more natural expression~$g^{-1}(S(\rho))$. This is the mean photon number of a Gaussian state with identical entropy as~$\rho$. It appears that at least for our isoperimetric inequality, such a generalization would require more than a na\"ive substitution  as there is no meaningful lower bound on the product~$g^{-1}(S(\rho))J(\rho)$ even for Gaussian states. While this may be considered as additional mathematical justification for our formulation of these inequalities, we believe that progress in this direction could be helpful in resolving, e.g., our conjecture concerning the
convergence rate to the fixed point of the quantum Ornstein-Uhlenbeck (qOU) semigroup. For the latter problem, apart from proving our conjecture, it would also be interesting to obtain multi-mode generalizations. This concerns, in particular, the entropy production rates for the qOU semigroup. Here, a resolution of the conjecture of~\cite{depalma2016}  for the multi-mode attenuator would likely provide important insights.

Finally, we mention some challenging mathematical problems resulting from our work. For example, while our isoperimetric inequality is tight for Gaussian states, necessary and sufficient conditions for equality in most of our statements are currently unknown. Finally, a rigorous discussion of the family of states for which the de Bruijn identity~\eqref{eq:de_Bruijn_heat} holds, possibly using the framework of Schwartz operators \cite{keylkiukaswerner15}, would be an interesting task for future work.

\section{Remark}
After posting our paper to the arxiv, we were made aware of concurrent related work by Rouz\'{e}, Datta, and Pautrat. Their paper has now been made available \cite{dattaetal16}.

\section*{Acknowledgements}
RK thanks the organizers of the workshop on Hypercontractivity and Log-Sobolev Inequalities in Banff and R. F. Werner for interesting
discussions. RK and SH are supported by the Technische Universit\"at M\"unchen –- Institute
for Advanced Study, funded by the German Excellence
Initiative and the European Union Seventh Framework
Programme under grant agreement no. 291763. They acknowledge additional support by DFP project no. K05430/1-1. AV is supported by the John Templeton Foundation Grant No. 48322, and by the Basque Government through the BERC 2014-2017 program and by Spanish Ministry of Economy and Competitiveness MINECO: BCAM Severo Ochoa excellence accreditation SEV-2013-0323.

\appendix

\section{A Log-Sobolev inequality and the classical Ornstein-Uhlenbeck process}
\label{sec:classicalOU}
In this appendix we discuss known classical results for the reader's convenience: we briefly review
the relationship between the
isoperimetric inequality for classical Fisher information,
the Log-Sobolev inequality, and 
the rate of convergence to the fixed point for the classical Ornstein-Uhlenbeck semigroup. These arguments were given by Carlen~\cite{carlen94}
for a particular element of the two-parameter family of  Ornstein-Uhlenbeck processes. Here we specialize to real-valued random variables, but allow arbitrary parameters in order to illustrate the parallels to the qOU semigroup. We also explicitly discuss the convergence to the fixed point,
which is only implicit in~\cite{carlen94} but appears to be folklore.

Let $f_0$ be a probability density on~$\mathbb{R}$ of a real-valued random variable~$X$. The classical Ornstein-Uhlenbeck (cOU) process, for $\theta>0$ and $\sigma>0$ is  given by the Fokker-Planck equation
\begin{align}
\frac{\partial f}{\partial t}=\theta \frac{\partial }{\partial x}\left[x f\right]+\frac{\sigma^2}{2}\frac{\partial^2 f}{\partial x^2}=\cA_{\theta,\sigma^2}(f)\ .
\end{align}
(Carlen considers the case where $\theta=1$ and $\sigma^2=1/\pi$.) 
The solution to this equation can be written (in terms of random variables) as
\begin{align}
X_t &=e^{-\theta t}X_0+\frac{\sigma}{\sqrt{2\theta}}\sqrt{1-e^{-2\theta t}}Z\ ,\label{eq:xtevolution}
\end{align}
where $Z\sim\cN(0,1)$ is an independent centered Gaussian random variable with unit variance.  The stationary solution~$f_\infty$ therefore is a  centered Gaussian distribution with variance~$\sigma^2/(2\theta)$.
In particular,~\eqref{eq:xtevolution}
implies that the second moments satisfy
\begin{align}
\mathbb{E}[X_t^2]&=e^{-2\theta t}\mathbb{E}[X_0^2]+
(1-e^{-2\theta t})\frac{\sigma^2}{2\theta}\ .\label{eq:momentsevolutionc}
\end{align}
We use that the relative entropy 
between a random variable~$X$ and a centered normal variable~$Z_{\sigma^2}\sim \cN(0,\sigma^2)$ is given by
$D(X\|Z_{\sigma^2}) = -H(X)+\frac{1}{2}\log 2\pi \sigma^2+\frac{1}{2\sigma^2}\mathbb{E}[X^2]$
as can be verified easily.
In particular, the relative entropy between the solution at time~$t$
and the fixed point~$Z_{\sigma^2/(2\theta)}$ is given by
\begin{align}
D(X_t\|Z_{\sigma^2/(2\theta)}) &=-H(X_t)+\frac{1}{2}\log \pi \sigma^2/\theta+\frac{\theta}{\sigma^2}\mathbb{E}[X_t^2]\ .\label{eq:relativeentropytofixed}
\end{align} 
Furthermore, we have the following de Bruijn-type identity:
\begin{lemma}
Let $\{X_t\}_{t\geq 0}$ be of the form~\eqref{eq:xtevolution}, i.e., a solution to the cOU process. 
Let 
$J(X)=\int \frac{|\frac{\partial f(x)}{\partial x}|^2}{f(x)} dx$
denote the Fisher information of a random variable~$X$ with distribution function $f$. Then
\begin{align}
\frac{d}{dt}\Big|_{t=0} H(X_t)&=\frac{\sigma^2}{2}J(X_0)-\theta\ .\label{eq:debruijinclassicalx}
\end{align}
\end{lemma}
Observe that the second summand essentially stems from the fact that
entropies transform very simply under rescaling of random variables, namely
\begin{align} 
H(e^{-\theta t}X)=H(X)-\theta t\ .\label{eq:entropyrescalingclassical}
\end{align}
Eq.~\eqref{eq:entropyrescalingclassical}  significantly simplifies the analysis. Such a property does not hold in the quantum case: as a consequence, we do not have a simple expression in terms of $J(X_0)$ only.

\begin{proof}
The derivative of the entropy along a semigroup with generator $\cA$ is given by the expression $\frac{d}{dt}\Big|_{t=0}H(e^{t\cA}(f))=-\int\cA(f)(x)\log f(x)dx$. Using this fact gives  
\begin{align}
\frac{d}{dt}\Big|_{t=0}H(X_t)&=
-\theta \int \left(\frac{\partial }{\partial x} [x f(x)]\right)\log f(x)dx-\frac{\sigma^2}{2}\int \frac{\partial^2 f(x)}{\partial x^2} \log f(x)dx\ .
\end{align}
Denoting $f'(x)=\frac{\partial}{\partial x}f(x)$, we obtain
\begin{align}
\int (x f)'\log f dx&=-\int (x f) f'/f dx=-\int x f'dx=\int f dx=1\ ,
\end{align}
where we have used partial integration and the fact that boundary terms vanish twice. Similarly,
we have
\begin{align}
\int f'' \log f dx&=-\int (f')^2/f dx
\end{align}
by partial integration. Combining these statements gives the claim. 
\end{proof}

Following~\cite{carlen94}, we can write the isoperimetric inequality $1/N(X)\leq \frac{J(X)}{2\pi e}$ as 
\begin{align}
-H(X)&\leq \frac{1}{2}\log \left(\frac{J(X)}{2\pi e}\right)\\
&=\frac{1}{2}\log \left(\frac{J(X)}{2\pi }\right)-\frac{1}{2}\ .
\end{align}
In particular, using
  $\log x\leq x-1$, we get 
\begin{align}
\log \left(\frac{J(X)}{2\pi e}\right)&=\log \left(\frac{J(X)\cdot A}{2\pi eA}\right)\nonumber\\
&=\log\frac{1}{2\pi e A} +\log( A J(X))\nonumber\\
&\leq AJ(X)-2-\log 2\pi  A\ ,\label{eq:hupperboundclassical}
\end{align}
for any $A>0$. 
Using inequality~\eqref{eq:hupperboundclassical}, it is straightforward to show the following.

\begin{theorem}[Fast convergence of the cOU semigroup~\cite{carlen94}]\label{thm:carlengeneralized}
Let $\{X_t\}_{t\geq 0}$ be of the form~\eqref{eq:xtevolution}, i.e., a solution to the cOU process
 with parameters~$\theta>0,\sigma>0$.
 Then 
\begin{align}
 \frac{d}{dt}\Big|_{t=0}D(X_t\|Z_{\sigma^2/(2\theta)})&\leq -2\theta  D(X_0\|Z_{\sigma^2/(2\theta)})\ .
\end{align}
\end{theorem}
\noindent Note that because we are considering a semigroup, this result
immediately implies that 
\begin{align}
D(X_t\|Z_{\sigma^2/(2\theta)})\leq e^{-2\theta t}D(X_0\|Z_{\sigma^2/(2\theta)}) \qquad \textrm{ for all }t\geq 0\ .
\end{align}
Also, this result is tight:
if $X_0\sim \cN(0,\sigma^2_{X_0})$ is a centered Gaussian random variable with variance~$\sigma^2_{X_0}$, then~\eqref{eq:xtevolution} implies $X_t\sim \cN(0,\sigma_t^2)$ where the variance of $X_t$ is 
\begin{align}
\sigma^2_t&=e^{-2\theta t}\sigma^2_{X_0}+\frac{\sigma^2}{2\theta}(1-e^{-2\theta t})\ .
\end{align}
Inserting into~\eqref{eq:relativeentropytofixed}, using $H(X_t)=\frac{1}{2}(1+\log(2\pi \sigma_t^2))$
yields
\begin{align}
D(X_t\|Z_{\sigma^2/(2\theta)})&=-\frac{1}{2}(1+\log(2\pi \sigma_t^2))+\frac{1}{2}\log \pi \sigma^2/\theta+\frac{\theta}{\sigma^2}\sigma_t^2\ .
\end{align}
In particular, we obtain
\begin{align}
\lim_{\sigma_{X_0}^2\rightarrow\infty}\left(\frac{d}{dt}\Big|_{t=0} D(X_t\|Z_{\sigma^2/(2\theta)})\right)/D(X_0\|Z_{\sigma^2/(2\theta)})=-2\theta\ .
\end{align}

\begin{proof}
According to Eqs.~\eqref{eq:momentsevolutionc},~\eqref{eq:relativeentropytofixed} and~\eqref{eq:debruijinclassicalx},  we have 
\begin{align}
\frac{d}{dt}\Big|_{t=0}
 D(X_t\|Z_{\sigma^2/(2\theta)})
 &=-\frac{\sigma^2}{2}J(X_0)-\frac{2\theta^2}{\sigma^2}\mathbb{E}[X^2]+2\theta\ .\label{eq:ddterivativeclassical}
\end{align} 
Combining~\eqref{eq:hupperboundclassical} with~\eqref{eq:relativeentropytofixed}
 yields
\begin{align}
 D(X_0\|Z_{\sigma^2/(2\theta)})\leq \frac{AJ(X)}{2}-(1+\frac{1}{2}\log 2\pi A)+\frac{1}{2}\log \pi\sigma^2/\theta+\theta/\sigma^2\mathbb{E}[X_0^2]\ .\label{eq:dxzerozsigma}
\end{align}
Combining~\eqref{eq:dxzerozsigma} with~\eqref{eq:ddterivativeclassical}
 yields 
 \begin{align}
 -2\theta  D(X_0\|Z_{\sigma^2/(2\theta)})-\frac{d}{dt}\Big|_{t=0}
 D(X_t\|Z_{\sigma^2/(2\theta)})&\geq \left(-\theta  A+\frac{\sigma^2}{2}\right) J(X_0)\\
& + 2\theta\left[\left(1+\frac{1}{2}\log 2\pi A-\frac{1}{2}\log\pi \sigma^2/\theta\right)-1\right]\ .\label{eq:logsobolevclassical}
 \end{align}
The claim then follows by choosing~$A=\frac{\sigma^2}{2\theta }$.
\end{proof}

\section{Tightness of the quantum Fisher information isoperimetric inequality}
\label{app:Fisher}
Consider a one-mode Gaussian thermal state $\omega_\n$ with mean photon number $\n > 0$. Its entropy is
\begin{align}
  \label{eq:entropygaussian}
S(\omega_\n)&=g(\n)=(\n+1)\log (\n+1)-\n\log \n\ .
\end{align}
Under the diffusion semigroup, the state $\omega_\n$ evolves as
\begin{align*}
e^{t\cLh}(\omega_\n)&=\omega_{\n_t}\ ,
\qquad\textrm{ where }\qquad \n_t=\n+2\pi t\ .
\end{align*}
In particular, by the de Bruijn identity
\begin{align}
  \label{eq:debruijnheat}
  J(\omega_\n) = 2\, \frac{d}{dt} S(e^{t\cLh}(\omega_\n))\bigg|_{t=0}=
2g'(\n)\n'_t\Big|_{t=0}=4\pi\log\left(\frac{\n+1}{\n}\right)\ .
\end{align}
Also,
\begin{align}
  J(e^{t\cLh}(\omega_\n))=4\pi\log\left(1+\frac{1}{\n+2\pi t}\right)\ .
\end{align}
Calculating the right-hand side of the quantum Fisher information isoperimetric inequality \eqref{eq:iso}, we obtain
\begin{align}
  \frac{d}{dt}\bigg|_{t=0}\Bigl[\frac{1}{2}J(e^{t\cLh}(\omega_\n))\Bigr]^{-1} &=\frac{1}{\n(\n+1)}\log^{-2}\left(1+\frac{1}{\n}\right)\rightarrow\ 1\qquad \text{ as }\n\rightarrow\infty\ .
\end{align}

\section{Tightness of the isoperimetric inequality}
\label{app:isoperim}
Consider a one-mode Gaussian thermal state $\omega_\n$ with mean photon number $\n$. From Eq.~\eqref{eq:debruijnheat} we know that
\begin{align}
  J(\omega_\n) = 
4\pi\log\left(\frac{\n+1}{\n}\right)\ ,
\end{align}
and that the entropy power of $\omega_\n$ is given by (cf. \eqref{eq:entropygaussian})
\begin{align}
  N(\omega_\n)= \exp(S(\omega_\n)/1) = \frac{(\n+1)^{\n+1}}{\n^\n}\ .
\end{align}
Combining these two expressions, we see that the left-hand side of \eqref{eq:iso_entropy} is
\begin{align}
  J(\omega_\n)N(\omega_\n) = 4\pi \left(\frac{\n+1}{\n}\right)^{\n}\log\left(\frac{\n+1}{\n}\right)^{\n+1}{\rightarrow} \ 4\pi e\qquad \text{ for }{\n \rightarrow \infty}\ .
\end{align}

\section{On the convergence rate for Gaussian initial states}
\label{app:qOU}

In this section we show that for a one-mode Gaussian state $\rho$ the following inequality holds:
\begin{align}
\frac{d}{dt}\Big|_{t=0} D(e^{t\cL_{\mu,\lambda}}(\rho)\|\sigma_{\mu,\lambda})\leq -\zeta D(\rho||\sigma_{\mu,\lambda})\qquad&
\text{ with }\zeta = \mu^2 - \lambda^2 > 0\ .
\label{eq:appineq}
\end{align}
Furthermore, this statement is optimal: for any $\epsilon>0$ there exists a Gaussian state $\rho$ such that 
\begin{equation}
\frac{d}{dt}\Big|_{t=0} D(e^{t\cL_{\mu,\lambda}}(\rho)\|\sigma_{\mu,\lambda})\geq -(\zeta+\epsilon) D(\rho||\sigma_{\mu,\lambda})\ .
\end{equation}

First, we note that the entropy of a Gaussian state does not depend on its first moments. Hence it follows that for a Gaussian state the right-hand side of Eq. \eqref{eq:rate_decay} does not depend on its first moments, and it is suffices to consider centered Gaussian states.

Calculating the right-hand side of \eqref{eq:rate_decay}, we focus on calculating $J_\pm(\rho)$. Recall that the covariance matrix $M$ of a centered state $\rho$ is defined as  $M_{jk}=\Tr(\rho\{R_j, R_k\})$.

\begin{lemma} Let $\rho$ be a one-mode centered Gaussian state with mean-photon number $\n$ and covariance matrix given by $M = \kappa S^TS$, where $\kappa=2\n+1$, and $S = O_1 \begin{pmatrix}z & 0 \\ 0 & 1/z \end{pmatrix} O_2^T$ with $O_i \in Sp(2) \cap O(2)$ and  $z \geq 1$. Denote $$\left. J_\pm(\rho)= 2\frac{d}{dt} S\left(e^{t\cL_\pm}(\rho) \right)\right|_{t=0}\ . $$ Then we have

\begin{equation}\label{eq:LJ_Gaussian}
 J_\pm(\rho)=\left(\half \left(1/z^2 + z^2\right) \pm \kappa\right) \left(\log\left(\frac{\kappa +1}{2}\right)-\log\left(\frac{\kappa-1}{2}\right) \right)\ .
\end{equation}
\end{lemma}
\begin{proof}
  The covariance matrix of the time-evolved state $e^{t\cL_{\pm}}(\rho)$ is
\begin{equation}
M_\pm(t)=c_1^\pm(t)M(0)+c_2^\pm(t)\id\ ,
\end{equation}
where 
\begin{align}
c_1^-(t)=e^{-t} \qquad &\text{ and }\qquad c_2^-(t)=1-e^{-t}\ ,\label{eq:c_-}\\
c_1^+(t)=e^{t} \qquad &\text{ and }\qquad c_2^+(t)=e^{t}-1\ .\label{eq:c_+}
\end{align}
Therefore writing $S[M]$ for the entropy of a Gaussian state with covariance matrix $M$, we have
\begin{align}
  S\left[M_\pm(t) \right] &= S\left[c_1^\pm(t) \kappa S^T  S + c_2^\pm(t) \id\right] = S\left[c_1^\pm(t) \kappa O_2^T K^T O_1^T  O_1 K O_2 + c_2^\pm(t) O_2^T O_2 \right] \\
  &= S\left[c_1^\pm(t)\kappa  K^T  O_1^T O_1 K + c_2^\pm(t) \id \right]\\
  &= S\left[\begin{pmatrix} c_1^{\pm}(t) \kappa z^2 + c_2^{\pm}(t) & 0 \\ 0 & c_1^{\pm}(t)\kappa/z^2 + c_2^{\pm}(t) \end{pmatrix}\right]\ .
  \label{eq:sympleval}
\end{align}
The symplectic eigenvalue of the matrix argument in \eqref{eq:sympleval} is the square root of its determinant:
\begin{equation}\label{eq:kappa}
  \kappa_\pm(t) = \sqrt{\left(c_1^\pm(t)\kappa + c_2^\pm(t) z^2\right)\left(c_1^\pm(t)\kappa + c_2^\pm(t)/z^2\right)}\ .
\end{equation}
The entropy of the time evolved state is
\begin{equation}\label{eq:entropy_Gaussian}
  S(e^{t\cL_{\pm}}(\rho)) = g(\n(\kappa_\pm(t)))\ ,
\end{equation}
where $g(x) = (x+1)\log(x+1) - x\log x$ and $\n(\kappa) = \frac{1}{2}(\kappa - 1)$. By the chain rule, we have that
\begin{equation}
  \frac{d}{dt} S\left(e^{t\cL_{\pm}}(\rho)\right) = \half g'(\n(\kappa_\pm(t)))\, {\kappa_\pm}'(t)\ .
\end{equation}
Since $g'(x) = \log(x+1) - \log x $, we only need to find ${\kappa_\pm}'(t)$.

Combining \eqref{eq:kappa} and \eqref{eq:c_-} we obtain
\begin{align}
  \kappa_-(t) &= \sqrt{\left(e^{-t}\kappa + \left(1 - e^{-t}\right)z^2\right)\left(e^{-t}\kappa + \left(1 - e^{-t}\right)/z^2\right)}\\
  &= \kappa + t \left(\half \left(z^2 + 1/z^2\right) - \kappa\right) + O\left(t^2\right)\ .
\end{align}
Therefore $\left.{\kappa_-}'(t)\right|_{t = 0} = \half \left(1/z^2 + z^2\right) - \kappa$ and finally

\begin{equation}
  J_-(\rho)= \left(\half \left(1/z^2 + z^2\right) - \kappa\right) \log\frac{\kappa +1}{\kappa-1}\ .
\end{equation}
Similarly, using \eqref{eq:kappa} together with \eqref{eq:c_+}, we obtain
\begin{align}
  \kappa_+(t) &= \sqrt{\left(e^t\kappa + \left(e^t - 1\right)z^2\right)\left(e^t\kappa + \left(e^t - 1\right)/z^2\right)}\\
  &= \kappa + t \left(\half \left(z^2 + 1/z^2\right) + \kappa\right) + O\left(t^2\right)\ .
\end{align}
Thus $\left.{\kappa_+}'(t)\right|_{t = 0} = \half \left(1/z^2 + z^2\right) + \kappa$ and
\begin{equation}
  J_+(\rho)=\left(\half \left(1/z^2 + z^2\right) + \kappa\right) \log\frac{\kappa +1}{\kappa-1}\ .
\end{equation}

\end{proof}

From Lemma \ref{lem:rate_decay} it is clear that we are interested in minimizing $J_\pm(\rho)$. Both expressions in \eqref{eq:LJ_Gaussian} have a minimum at $z=1$. Therefore, with $\zeta = \mu^2 - \lambda^2$ and $\n = \frac{\kappa - 1}{2}$, from Lemma \ref{lem:rate_decay} we obtain

\begin{align}
-\zeta D(\rho||\sigma_{\mu,\lambda})-\frac{d}{dt}\Big|_{t=0} D(e^{t\cL_{\mu,\lambda}}(\rho)\|\sigma_{\mu,\lambda})\geq &\mu^2\log(\n+1) - \lambda^2\log\n + \lambda^2\log\lambda^2\\
    &- \mu^2\log\mu^2 + (\mu^2-\lambda^2)\log(\mu^2-\lambda^2)\\
    =: &h_{\mu,\lambda}(\n)\ .
\end{align}

For fixed $\mu^2 > \lambda^2$, the function $h_{\mu,\lambda}$ satisfies $\lim_{\n\rightarrow 0} h_{\mu,\lambda}(\n) = \lim_{\n\rightarrow \infty} h_{\mu,\lambda}(\n) = \infty$. Since $h_{\mu,\lambda}$ is differentiable (in fact smooth) for $\n > 0$, we can find the global minimum by finding the zeros of the derivative
\begin{align}
  \frac{d}{d\n} h_{\mu,\lambda}(\n) = \frac{\mu^2}{\n+1}-\frac{\lambda^2}{\n} = 0\ .
\end{align}
The only solution is $\n = \frac{\lambda^2}{\mu^2-\lambda^2}$, and since 
\begin{align}
  \frac{d^2}{d\n^2} \bigg|_{\n = \frac{\lambda^2}{\mu^2-\lambda^2}} h_{\mu,\lambda}(\n) = (\mu^2-\lambda^2) \left(\frac{1}{\lambda^2} - \frac{1}{\mu^2}\right) > 0\ ,
\end{align}
it is the minimum. The value of the minimum is $h_{\mu,\lambda}\left(\frac{\lambda^2}{\mu^2-\lambda^2}\right) = 0$, hence it follows that $h_{\mu,\lambda}(\n) \geq 0$ for all $ \n > 0$ and we have

\begin{align}\label{eq:conv_tight}
 \frac{d}{dt}\bigg|_{t=0} D(e^{t\cL_{\mu,\lambda}}(\rho)\|\sigma_{\mu,\lambda}) \leq - \zeta D(\rho\|\sigma_{\mu,\lambda})\ .
\end{align}
Moreover, the last inequality becomes equality for $\rho = \sigma_{\mu,\lambda}=\omega_{\n_\infty}$, with  $\n_\infty= \frac{\lambda^2}{\mu^2-\lambda^2}$.

It remains to show that $\zeta = \mu^2-\lambda^2$ is optimal. Let $\epsilon > 0$ and $\zeta' = \zeta + \epsilon$. Then for the Gaussian thermal state $\omega_\n$ we have

\begin{align}
  - \zeta' D(\omega_\n\| \sigma_{\mu,\lambda})-\frac{d}{dt}\bigg|_{t=0} D(e^{t\cL_{\mu,\lambda}}(\omega_\n)\|\sigma_{\mu,\lambda}) = &h_{\mu,\lambda}(\n) + \epsilon (\n+1)\log(\n+1) -\epsilon\,  \n\log\n \\
     &+ \epsilon\, \n \log\left(\frac{\lambda^2}{\mu^2}\right) + \epsilon\log\left(1-\frac{\lambda^2}{\mu^2}\right) \\
     =&\log\left( \left(\frac{\n+1}{\n}\right)^{\mu^2+\epsilon(\n+1)}\n^{\mu^2-\lambda^2+\epsilon}\left(\frac{\lambda^2}{\mu^2}\right)^{\epsilon\n} \right)\\
     & + \mathrm{c}(\mu,\lambda) \ {\rightarrow} - \infty \qquad \text{ for }{\n \rightarrow \infty}\ ,
\end{align}
where $\mathrm{c}(\mu,\lambda)=\log\left((\mu^2)^{-\mu^2-\epsilon}(\lambda^2)^{\lambda^2}(\mu^2-\lambda^2)^{\mu^2-\lambda^2+\epsilon} \right)$. 
Therefore, for any $\epsilon>0$, there exists $\n$ such that 
\begin{equation}
  \frac{d}{dt}\bigg|_{t=0} D(e^{t\cL_{\mu,\lambda}}(\omega_\n)\|\sigma_{\mu,\lambda}) > -(\zeta+\epsilon) D(\omega_\n||\sigma_{\mu,\lambda})\ .
\end{equation}
This shows that the constant $\zeta=\mu^2-\lambda^2$ is optimal in the inequality \eqref{eq:appineq}.


\end{document}